\pdfoutput=1
\RequirePackage{latexml}

\PassOptionsToPackage{
	pagebackref
}{hyperref}

\iflatexml
	\documentclass{article}
	\newenvironment{wrapstuff}[2][]{%
		\par%
		\begin{figure}[htbp]%
			\centering
		}{%
		\end{figure}%
		\par%
	}
	\usepackage{natbib}
	\usepackage[colorlinks]{hyperref}
	\usepackage{amssymb,amsthm,tikz}
	\newtheorem{theorem}{Theorem}[section]
	
	\newtheorem{proposition}[theorem]{Proposition}
	\newtheorem{corollary}[theorem]{Corollary}
	\newtheorem{definition}[theorem]{Definition}
	\newtheorem{example}[theorem]{Example}
	
\else
	
	\documentclass[format=acmsmall, nonacm=true, screen]{acmart}
	
	\acmConference{}{}{}
	\renewcommand{\footnotetextcopyrightpermission}[1]{} 
	\renewcommand{\footnotetextauthorsaddresses}[1]{} 
	
\fi

\usepackage{amsmath}
\usepackage{booktabs}
\usepackage{enumitem}
\usepackage{microtype}
\usepackage{thm-restate}

\title{Committee Monotonicity and Proportional Representation for Ranked Preferences}
\date{}

\usepackage[T1]{fontenc}
\usepackage{xspace}
\usepackage{graphicx}
\usepackage{xcolor}
\usepackage[ruled,linesnumbered, noend]{algorithm2e}
\usepackage{tikz}
\usetikzlibrary{fit,positioning}
\usetikzlibrary{backgrounds,matrix}
\usetikzlibrary{shapes.geometric}
\usepackage{pgfplots}
\usepackage{subcaption}
\usepackage[nameinlink]{cleveref}
\iflatexml\else
	\usepackage{crossreftools}
\fi
\usepackage{etoolbox}

\usepackage{subcaption}
\newcommand{\jr}{JR\xspace}
\newcommand{\pjr}{PJR\xspace}
\newcommand{\pjrp}{PJR+\xspace}
\newcommand{\rjr}{Rank-\jr}
\newcommand{\rpjr}{Rank-\pjr}
\newcommand{\rpjrp}{Rank-\pjrp}
\newcommand{\psc}{PSC\xspace}
\newcommand{\ipsc}{IPSC\xspace}

\DeclareMathOperator{\rank}{rank}

\DeclareMathOperator{\swap}{swap}
\newcommand{\periphery}[2]{\textup{periph}_{#2}(#1)}
\newcommand{\spref}{\succ}
\let\oldrhd\rhd
\iflatexml
	\newcommand{\rhdsim}{\mathbin{\raise0ex\hbox{\ooalign{\hss$\oldrhd$\hss\cr%
				\kern0.6ex\raise-0ex\hbox{\kern-0.6ex$\sim$}}}\kern-0.2ex}}
\else
	\newcommand{\rhdsim}{\mathbin{\raise0.15ex\hbox{\ooalign{\hss$\oldrhd$\hss\cr%
				\kern0.6ex\raise-0.82ex\hbox{\kern-0.6ex$\sim$}}}\kern-0.2ex}}
\fi
\renewcommand{\rhd}{\mathbin{\oldrhd\kern-0.2ex}}

\usepackage{mathrsfs}

\usepackage{xcolor}
\usepackage{ifthen}
\definecolor{alt-a}{HTML}{ffadad}
\definecolor{alt-b}{HTML}{ffd6a5}
\definecolor{alt-c}{HTML}{cdd1fa}
\definecolor{alt-c'}{HTML}{cdd1fa}
\definecolor{alt-x}{HTML}{fad4d7}
\definecolor{alt-d}{HTML}{fad4d7}
\definecolor{alt-e}{HTML}{b9d2fa}
\definecolor{alt-f}{HTML}{cbc2ff}
\definecolor{alt-g}{HTML}{ffc6ff}
\newcommand{\alternative}[2][default]{%
	\tikz[transform shape,scale=0.85]{
		\def\colorparam{#1}
		\ifthenelse{\equal{#1}{default}}{
			\def\colorparam{alt-#2}
		}{}
		\clip (0,3pt) circle (6pt);
		\node[anchor=base, yshift=0.5pt, text height=6pt,inner sep=1pt,fill=\colorparam, circle,minimum width=12pt] {};
		\node[anchor=base, text height=6pt, inner sep=1pt,minimum width=12pt,scale=1.15] {$#2$};
	}
}
\newcommand{\borderalternative}[2][default]{%
	\tikz[transform shape,scale=0.85]{
		\def\colorparam{#1}
		\ifthenelse{\equal{#1}{default}}{
			\def\colorparam{alt-#2}
		}{}
		\node[anchor=base, yshift=0.5pt, text height=6pt,inner sep=1pt,fill=\colorparam, circle,minimum width=12pt,draw=black!70,very thin] {};
		\node[anchor=base, text height=6pt, inner sep=1pt,minimum width=12pt,scale=1.15] {$#2$};
	}
}
\newcommand{\weakorder}[2][noborder]{%
	\begin{tikzpicture}
		[yscale=0.5]
		\foreach \row [count=\y] in {#2} {
			\ifthenelse{\y>1 \OR \equal{#1}{noborder}}{
				\node [anchor=south, inner sep=0.5pt] at (0, -\y) {
					\tikz[xscale=0.37]{
						\foreach \val/\color [count=\x] in \row {
							\ifthenelse{\equal{\val}{\color}}
							{\def\color{default}} %
							{}
							\node [inner sep=0] at (\x, 0) {\alternative[\color]{\val}};
						};}
				};
			}{
				\node [anchor=south, inner sep=0.5pt] at (0, -\y) {
					\tikz[xscale=0.37]{
						\foreach \val/\color [count=\x] in \row {
							\ifthenelse{\equal{\val}{\color}}
							{\def\color{default}} %
							{}
							\node [inner sep=0] at (\x, 0) {\borderalternative[\color]{\val}};
						};}
				};
			}
		}
	\end{tikzpicture}
}
\newcommand{\votermultiplicity}[2]{%
	\begin{tikzpicture}
		\node[anchor=base,inner sep=1pt] at (0,0) {#1};
		\node[anchor=north,inner sep=0] at (0,-0.2) {#2};
	\end{tikzpicture}
}
\pgfplotsset{
	shortl/.style={%
		legend image code/.code={
			\draw[##1,line width=1.6pt]
			plot coordinates {
				(0cm,0cm)
				(0.2cm,0cm)
			};%
		}
	},
}

\allowdisplaybreaks

\DeclareMathOperator*{\argmax}{arg\,max}

\iflatexml
\author{Haris Aziz \\ UNSW Sydney \and
	Patrick Lederer \\ UNSW Sydney \and
	Dominik Peters \\ CNRS, LAMSADE, Universit\'e Paris Dauphine - PSL \and
	Jannik Peters \\ National University of Singapore \and
	Angus Ritossa \\ UNSW Sydney
	}
\else
	\usepackage{wrapstuff}
	\setcitestyle{authoryear}

	\author{Haris Aziz}
	\email{haris.aziz@unsw.edu.au}
	\affiliation{\institution{UNSW Sydney} \country{Australia}}
	
	\author{Patrick Lederer}
	\email{p.lederer@unsw.edu.au} 
	\affiliation{\institution{UNSW Sydney} \country{Australia}}
	
	\author{Dominik Peters}
	\email{dominik.peters@lamsade.dauphine.fr}
	\affiliation{\institution{CNRS, LAMSADE, Universit\'e Paris Dauphine - PSL} \country{France}}
	
	\author{Jannik Peters}
	\email{peters@nus.edu.sg}
	\affiliation{\institution{National University of Singapore} \country{Singapore}}
	
	\author{Angus Ritossa}
	\email{a.ritossa@student.unsw.edu.au}
	\affiliation{\institution{UNSW Sydney} \country{Australia}}
\fi

\begin{document}

\begin{abstract}
We study committee voting rules under ranked preferences, which map the voters' preference relations to a subset of the alternatives of predefined size. In this setting, the compatibility between proportional representation and committee monotonicity is a fundamental open problem that has been mentioned in several works. We address this research question by designing a new committee voting rule called the Solid Coalition Refinement (SCR) rule that simultaneously satisfies committee monotonicity and Dummett's Proportionality for Solid Coalitions (\psc) property as well as one of its variants called inclusion PSC. This is the first rule known to satisfy both of these properties. Moreover, we show that this is effectively the best that we can hope for as other fairness notions adapted from approval voting are incompatible with committee monotonicity. For truncated preferences, we prove that the SCR rule still satisfies PSC and a property called independence of losing voter blocs, thereby refuting a conjecture of \citet{GJM24a}. Finally, we discuss the consequences of our results in the context of rank aggregation.
\end{abstract}

\maketitle
	
\vspace{1cm}
\setcounter{tocdepth}{2}
\tableofcontents
	
\newpage
\section{Introduction}

The Single Transferable Vote (STV) is a voting rule that allows the selection of multiple winners in an election. Voters provide rankings of the candidates that are running, and STV selects winners following the principles of proportional representation, which ensures that each group of voters has an influence approximately proportional to its size.
On a high level, STV is based on a minimum number of voters that a candidate needs to be elected (the so-called \emph{quota}). A candidate who is placed first by a number of voters exceeding the quota is elected; if not enough such candidates exist, the candidates with the fewest first-place votes are eliminated until additional candidates can be elected.
STV in its multi-winner proportional representation version was proposed as early as 1819, and is used for electing national legislatures in Australia, Ireland, and Malta, as well as local governments in New Zealand, Scotland, and Northern Ireland \citep{Tideman95}.

Recent literature in computational social choice has noted that proportional representation has applications far beyond parliament elections, such as in budgeting problems \citep{PPS21a,ReyPaBuSurvey2023}, recommender systems \citep{LuBo11d,SFL16a}, multi-criteria decision making \citep{LPW24a}, and training AIs based on collective preferences \citep{Pete24a}. A rule like STV could thus be useful in these applications. However, STV exhibits some undesirable behaviors when the number of candidates to be elected is changed. In particular, if the number of winners is increased, a candidate that was previously among the winners may now be designated as a loser.
In other words, STV violates \emph{committee monotonicity}, which requires a voting rule to select a superset of the current winners if the number of winners is increased  \citep{Elkind2017}. Unfortunately, this makes STV unsuitable for many 
applications of proportionally representative voting, where committee monotonicity is often a requirement because the number of winners is uncertain or might dynamically change.
Let us consider some examples.
\begin{itemize}
	\item A lecturer lets students vote over which topics will be covered in their course. Proportionality is desirable in this context because topics should be chosen so as to keep many students interested. The lecturer expects to be able to cover 6 topics, say, but the class may proceed more quickly than expected, so the lecturer needs the flexibility to expand coverage to 7--8 topics. A similar application concerns interactive Q\&A systems \citep{IsBr24a}, where the audience submits and votes on questions during a talk or panel discussion. Since it is hard to predict how many of the questions will be asked, committee monotonicity is desirable.
	\item A multi-criteria recommender system makes recommendations based on modelling the user's preferences along several dimensions. For example, it might rank hotels based on room size, price, location, etc., with the user specifying their relative importance. The recommender system will recommend 10 hotels, say, by combining these criteria following the importance weight via proportional representation. The users can click a ``show more'' button that will display 10 more hotels -- without thereby wanting to disqualify any of the initial 10 hotels.
	\item An award committee votes over who should receive awards. The committee has some discretion over the number of awards to hand out, depending on the strength of the nominees. Making the final decision is simplified if the voting rule is committee monotone. Proportionality is frequently desirable in these situations to cover different types of achievement.
\end{itemize}

\begin{figure}[t]
	\definecolor{alt-a}{HTML}{f6b395}
	\definecolor{alt-b}{HTML}{fddbc7}
	\definecolor{alt-c}{HTML}{fbe8d8}
	\definecolor{alt-d}{HTML}{f8f6e9}
	\definecolor{alt-e}{HTML}{e5eeec}
	\definecolor{alt-f}{HTML}{d1e5f0}
	\definecolor{alt-g}{HTML}{9bc6df}
	\centering
	\iflatexml
	\includegraphics{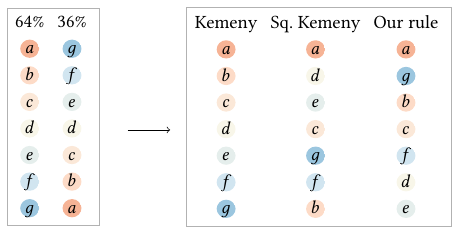}
	\else
	\scalebox{0.9}{
	\begin{tikzpicture}
		\node at (0,0) [draw=black!30] (step1) {	
			\votermultiplicity{64\%}{\weakorder{{a},{b},{c},{d},{e},{f},{g}}}
			\votermultiplicity{36\%}{\weakorder{{g},{f},{e},{d},{c},{b},{a}}}
		};
		
		\draw[->] (1.4,-0.25) -- (2.2,-0.25);
		
		\node at (5,0) [draw=black!30] (step2) {	
			\votermultiplicity{Kemeny}{\weakorder{{a},{b},{c},{d},{e},{f},{g}}}
			\votermultiplicity{Sq. Kemeny}{\weakorder{{a},{d},{e},{c},{g},{f},{b}}}
			\votermultiplicity{Our rule}{\weakorder{{a},{g},{b},{c},{f},{d},{e}}}
		};
	\end{tikzpicture}}
	\fi
	\caption{An example showing the behavior of the Kemeny rule, the Squared Kemeny rule, and our Solid Coalition Refinement rule that is based on the D'Hondt method and which produces a ranking satisfying PSC. The Kemeny rule simply returns the ranking that is returned by a majority of voters and is thus not proportional at all. The Squared Kemeny rule aims to produce an average ranking of the (weighted) input rankings, in terms of swap distance. Hence, it is proportional with respect to swap distance. However, it is not proportional with respect to its top-ranked alternatives: for example, one would expect that candidate $g$ appears within the first three alternatives in a proportional ranking (since it is top-ranked by more than a third of the voters). Our method leads to a ranking that is intuitively proportional.}
	\label{fig:kem}
\end{figure}

Another application of committee monotone rules is \emph{rank aggregation} which is the task of combining several rankings into one. This is the quintessential social choice problem, going back to \citet{Arrow1950}. Note that a committee monotone rule implicitly computes an output ranking of the candidates, whose top-$k$ candidates are the candidates selected by the rule when asked for $k$ winners. Thus, a committee monotone rule satisfying proportionality criteria can be used to generate \emph{proportional rankings} \citep{SLB+17a}, a task that has received increasing attention \citep{LPW24a,IsBr24a}. Proportional rankings can be used, for example, to produce university rankings by aggregating rankings based on different criteria such as research output and teaching performance. Standard rank aggregation methods like the Kemeny rule are majoritarian and thus not appropriate when proportionality is desired. The recently-studied Squared Kemeny rule \citep{LPW24a} behaves like an average and can thus be seen as providing a proportional outcome. However, its proportionality characteristics apply to the global ranking, and there is no guarantee that the top-$k$ candidates proportionally reflect the voters.
In many contexts, however, it is natural to place a greater emphasis on the top candidates in a ranking. Committee voting rules that are committee monotone and return proportional committees are a good fit in those cases, as every prefix of their computed ranking corresponds to a proportional committee. See \Cref{fig:kem} for an example where Kemeny and Squared Kemeny are not proportional in this sense.

As we mentioned, STV cannot be used for these applications, but STV is not the only proportional voting method. Others have been proposed more recently, such as the Quota Borda System (QBS) \citep{Dummett1984} and the Expanding Approvals Rule (EAR) \citep{Aziz2019}. These rules all satisfy strong proportionality properties. However, strikingly, these rules also all fail committee monotonicity. Indeed, as \citet{LPW24a} point out, it is an open question if there exists any rule combining committee monotonicity with strong notions of proportional representation. The fundamental reason that existing methods fail is that they are internally based on a \emph{quota}, which intuitively means that a candidate is a winner if more than $\frac{n}{k+1}$ voters support the candidate, where $n$ is the total number of voters and $k$ is the number of winners. However, if $k$ changes, the quota changes and thereby the identities of the winning candidates may change in unpredictable ways.

In this paper, we will design the first committee voting rules for ranked preferences that satisfy both proportional representation and committee monotonicity. 
Given the problems faced by rules defined using quotas, this requires developing entirely new approaches for constructing proportional voting rules. Let us specify our goal more precisely.
We focus on voters submitting ranked preferences, potentially with indifferences (i.e., weak orders). We will design new \emph{committee voting rules}, which are functions that map the voters' preferences and a desired number of winners to a subset of the candidates of that size.
For proportional representation, we focus on variants of Dummett's \textit{Proportionality for Solid Coalitions (\psc)} \citep{Dummett1984}. This is the most common formalization of proportionality for ranked preferences, and has been referred to as ``a \textit{sine qua non} for a fair election rule'' \citep{Woodall1994}. It provides representation guarantees for groups of voters who put the same candidates at the top of their rankings, and it has been shown that STV and the other rules we mentioned satisfy \psc. Our formal goal is thus to design committee voting rules that satisfy both \psc and committee monotonicity.

\begin{figure}[t!]
	\centering
	\resizebox{0.7\linewidth}{!}{%
		\begin{tikzpicture}[node distance=1.2cm, auto]
			
			\node (droop) {D-PSC};
			\node [below of=droop] (hare) {H-PSC};
			
			\node [right of=droop, node distance=2.4cm] (pjrplus) {Rank-PJR+};
			\node [below of=pjrplus] (pjr) {Rank-PJR};
			\node [below of=pjr] (jr) {Rank-JR};
			\definecolor{darkgreen}{rgb}{0.0, 0.8, 0.0}
			
			\node[draw, dotted, line width=0.5mm, dash pattern=on 3pt off 3pt, draw=darkgreen, fit=(droop) (hare), inner sep=0.15cm] {};
			\node[above of=droop, anchor=south, node distance=0.45cm, align=center] (label) {Compatible\\Thm. \ref{thm:revseq}};
			
			\node[draw, dotted, line width=0.5mm, dash pattern=on 3pt off 3pt, draw=red!80!black, fit=(pjrplus) (pjr) (jr), inner sep=0.15cm] {};
			\node[above of=pjrplus, anchor=south, node distance=0.45cm, align=center] (label) {Incompatible\\Thm. \ref{thm:rjr}};

			\draw[->, line width=0.4mm] (droop) -- (hare);
			\draw[->, line width=0.4mm] (pjrplus) -- (pjr);
			\draw[->, line width=0.4mm] (pjr) -- (jr);
			\draw[->, line width=0.4mm] (pjr) -- (hare);
			
			\node [right of=pjrplus, node distance=3cm] (idroop) {D-IPSC};
			\node [right of=idroop, node distance=2cm] (ihare) {H-IPSC};
			\node [below of=idroop] (gdroop) {D-GPSC};
			\node [below of=ihare] (ghare) {H-GPSC};
			
			\node [right of=ihare, node distance=2.4cm] (wpjrplus) {Rank-PJR+};
			\node [below of=wpjrplus] (wpjr) {Rank-PJR};
			\node [below of=wpjr] (wjr) {Rank-JR};
			
			\node[draw, dotted, line width=0.5mm, dash pattern=on 3pt off 3pt, draw=darkgreen, fit=(idroop) (ihare) (gdroop) (ghare), inner sep=0.15cm] (dottedbox) {};
			\node[above of=idroop, anchor=south, node distance=0.45cm, align=center, xshift=1cm] (label) {Compatible\\Thm. \ref{thm:SCR-ipsc-and-monotone}};
			
			\node[draw, dotted, line width=0.5mm, dash pattern=on 3pt off 3pt, draw=red!80!black, fit=(wpjrplus) (wpjr) (wjr), inner sep=0.15cm] {};
			\node[above of=wpjrplus, anchor=south, node distance=0.45cm, align=center] (label) {Incompatible\\Thm. \ref{thm:rjr}};
			
			\draw[->, line width=0.4mm] (idroop) -- (ihare);
			\draw[->, line width=0.4mm] (gdroop) -- (ghare);
			\draw[->, line width=0.4mm] (ihare) -- (ghare);
			\draw[->, line width=0.4mm] (idroop) -- (gdroop);
			
			\draw[->, line width=0.4mm] (wpjrplus) -- (ihare);
			\draw[->, line width=0.4mm] (wpjrplus) -- (wpjr);
			\draw[->, line width=0.4mm] (wpjr) -- (ghare);
			\draw[->, line width=0.4mm] (wpjr) -- (wjr);
			
			\draw[black!50, line width=0.6mm] (pjrplus) ++(1.6cm,1.9cm) -- ++(0,-4.85cm);
			
			\node at(1.5, 1.7) {\Large Strict Preferences};
			\node at(8, 1.7) {\Large Weak Preferences};
		\end{tikzpicture}
	}
	\vspace{-7pt}
	\caption{A summary of our results. Axioms in the green boxes are compatible with committee monotonicity, whereas the axioms in the red boxes are incompatible. An arrow from one axiom to another means that the first implies the second.
		``D-'' is short for Droop and ``H-'' for Hare.
		Definitions for GPSC (Generalized \psc), \rpjr, and \rpjrp as well as proofs of the relationships between the axioms can be found in the work of \citet{BrPe23a}.
	}
	\label{fig:results}
\end{figure}

\paragraph{Contributions.} 
We begin our work by focussing on strict preferences (linear orders without indifferences). For this setting, we show that any voting rule satisfying \psc can be bootstrapped into a committee monotone voting rule satisfying \psc by running it in a ``reverse sequential'' mode (\Cref{thm:revseq}). In this mode, we repeatedly use the base rule to delete the worst candidate, namely the candidate that is not declared a winner by the base rule when the number of winners is one less than the number of available candidates.
Based on this result, we can construct an entire family of committee monotone voting rules that satisfy \psc, for example based on STV or EAR. However, this family of rules is rather technical as we need to repeatedly compute existing rules to derive the winning committee. This also means that a large amount of computation is necessary to determine the winners for these rules, especially if the number of desired winners is small.

We hence design another committee monotone rule, the so-called Solid Coalition Refinement (SCR) rule, that satisfies \psc while working more directly. This rule is inspired by the D'Hondt apportionment method, and works by identifying virtual parties in the preferences (which correspond to sets of candidates that are frequently ranked on top together). In more detail, analogous to the D'Hondt method of apportionment, the SCR rule repeatedly adds a candidate that represents the most underrepresented group of voters to the winning committee. Moreover, since there can be multiple such candidates, the exact candidate is determined through a process of refinement. As we will show, this rule satisfies committee monotonicity and \psc, and it can be computed in polynomial time if the voters report strict preferences. Even more, the SCR rule is well-defined for weak preferences and satisfies a generalization of \psc called inclusion PSC \citep{azizlee2021} for this setting (\Cref{thm:SCR-ipsc-and-monotone}). We note here that reverse sequential rules do not satisfy inclusion \psc when voters have weak preferences. Since the SCR rule works for weak orders, it is also applicable on the domain of approval votes (dichotomous preferences), on which it satisfies the well-known PJR axiom as a consequence of inclusion PSC. This is the strongest proportionality guarantee that is known to be compatible with committee monotonicity on the approval domain. However, our rule is the first rule that generalizes to weak orders: the existing committee monotone approval-based rules due to Phragm\'en fail PSC when generalized in the natural way \citep{Jans16a}.

In order to be committee monotone, the SCR rule does not make use of quotas. As a result, it avoids some other paradoxes that implicitly occur due to quota dependence. In particular, \citet{GJM24a} studied a model with truncated preferences, where voters strictly rank only some of the candidates. They noticed that in this setting, all known \psc method fail a property they call independence of losing voter blocs. This property requires that the outcome of a rule should not change if we delete voters who do not rank any winning candidates. Intuitively, existing \psc methods fail because deleting a voter bloc changes the value of the quota. \citet{GJM24a} hypothesized that \psc and independence of losing voter blocs are incompatible.
However, we show that the SCR rule satisfies both \psc and independence of losing voter blocs for truncated preferences (\Cref{thm:truncated}), thus disproving their conjecture. 

Moreover, we examine the compatibility of committee monotonicity with a family of proportionality notions due to \citet{BrPe23a} that adapt fairness axioms from approval-based committee voting to ranked preferences. However, it turns out that committee monotonicity is even incompatible with \rjr, the weakest such proportionality notion (\Cref{thm:rjr}). This demonstrates a striking difference between \psc and the new proportionality notions by \citeauthor{BrPe23a}, and it suggests that it may be impossible to attain stronger proportionality conditions than \psc with committee monotone voting rules.

Finally, we analyze the consequences of our results for rank aggregation, where the outcome is a ranking of the candidates instead of a committee. To this end, we first note that every committee monotone committee voting rule that satisfies \psc can be seen as a rank aggregation rule that guarantees that each prefix of the chosen ranking satisfies \psc. Moreover, as an alternate concept of proportional representation, we investigate the maximum swap distance between an input ranking and a ranking satisfying \psc, as a function of the fraction of voters that report the input ranking. This approach was recently suggested by \citet{LPW24a} to argue in favor of the Squared Kemeny rule. Specifically, we show in \Cref{thm:swap} that rankings satisfying \psc give comparable or even slightly better bounds than those chosen by the Squared Kemeny rule with respect to the maximum swap distance to an input ranking, which further justifies our approach.

\section{Related Work}

We next give a brief overview of prior works on proportional representation and committee monotonicity for committee voting with ranked preferences.

\paragraph{Proportional Representation.}
The problem of finding representative committees has a long tradition as the {Single Transferable Vote (STV)} was already proposed in the 19th century  \citep{Tideman95}. This rule is now widespread and used for elections in many countries, and it has been shown that STV provides proportional representation since it satisfies \psc \citep{Woodall1994,TidemanBetterVoting2000}. As a consequence of its widespread use, many works have focused on better understanding STV \citep[e.g.,][]{BartholdiSTV1991,Elkind2017,DePe24a}.
Moreover, numerous other committee voting rules have been suggested with the aim of finding representative outcomes. Examples include the Chamberlin--Courant rule \citep{ChCo83a,LuBo11d}, Monroe's rule \citep{Monr95a}, and various forms of positional scoring rules \citep{FSST19a}. However, these rules do not satisfy \psc. Other rules than STV that satisfy this condition include the \emph{Quota Borda System (QBS)} of \citet{Dummett1984} and the \emph{Expanding Approvals Rule (EAR)} of \citet{Aziz2019}. 
\citet{Dummett1984} suggests that QBS may be less chaotic than STV, while \citet{Aziz2019} motivate EAR by the observation that it satisfies monotonicity conditions that STV fails. Moreover, \citet{BrPe23a} show that EAR satisfies a fairness condition called \rpjrp which is violated by STV. 
Finally, \citet{AZIZ2022248} characterize committees satisfying \psc in terms of \emph{minimal demand rules}; however, it is not straightforward to derive appealing rules from this characterization. 
Rules based on pairwise comparisons between committees have also been discussed \citep{Schu18a,Tideman95,TidemanBetterVoting2000}, though it is unclear if they satisfy PSC.

\paragraph{Committee Monotonicity.}
Just like proportional representation, committee monotonicity (called \emph{house monotonicity} in other settings) was identified early on as a desirable property of committee elections. 
For instance, in 1880, the chief clerk of the census office of the United States noticed that Alabama would be allocated 8 seats in a 299-seat parliament but only 7 seats in a 300-seat parliament, a paradox which is today known as the ``Alabama Paradox'' 
\citep{FairRepresentationBook}. 
It is known that committee monotonicity and proportional representation are compatible in the simpler setting of (approval-based) apportionment \citep{FairRepresentationBook,BGP+24a}, where voters vote for parties that can be assigned multiple seats. By contrast, the analysis of committee monotonicity for committee elections consists largely of counterexamples showing that specific classes of rules fail this axiom \citep{Staring1986,Ratliff2003,BaCo08a,Kamwa2013,Elkind2017, McGr23a}. 
For instance, \citet{Elkind2017} present a counterexample showing that STV fails committee monotonicity.
On the other hand, \citet{Jans16a} showed that several rules, such as Phragm\'en's Ordered Method and Thiele's Ordered Method, satisfy committee monotonicity. However, these rules fail \psc. Finally, as mentioned before, voting rules that output rankings can be seen as committee monotone committee voting rules. Hence, the study of proportional rank aggregation rules \citep[][]{SLB+17a,LPW24a} is related to our work, but these papers focus on proportionality notions specific to rankings.

\section{Preliminaries}
Let $N=\{1,\dots, n\}$ denote a set of $n$ voters and let $C=\{c_1,\dots, c_m\}$ denote a set of $m$ candidates. We assume that every voter $i\in N$ reports a \emph{(weak) preference relation $\succsim_i$} over the candidates, which is formally a complete and transitive binary relation on $C$.
The notation $c \succsim_i c'$ denotes that voter $i$ weakly prefers candidate $c$ to $c'$. We write $c \succ_i c'$ if $c \succsim_i c'$ but not $c' \succsim_i c$, which indicates a strict preference.
We call a preference relation \emph{strict} if it is antisymmetric, i.e., there is no indifference between any two candidates. 
The set of all weak preference relations is denoted by $\mathcal R$ and the set of all strict preference relations by $\mathcal{L}$. 
A \emph{preference profile} $R$ is the collection of the voters' preferences, i.e., it is a function from $N$ to $\mathcal{R}$. A preference profile is \emph{strict} if all voters have strict preference relations. The set of all preference profiles is $\mathcal{R}^N$ and the set of all strict preference profiles is $\mathcal{L}^N$. 

Given a preference profile, our goal is to select a \emph{committee}, which is formally a subset of the candidates of a given size $k$. To this end, we use \emph{committee voting rules}, which for every preference profile and committee size $k$ return a committee of that size. More formally, a committee voting rule $f$ for strict (resp. weak) preferences maps every profile $R\in \mathcal{L}^N$ (resp. $\mathcal{R}^N$) and target committee size $k\in \{1,\dots, m\}$ to a winning committee $W=f(R,k)$ with $|W|=k$. Thus, committee voting rules always choose a single committee and so are \emph{resolute}.

We next introduce committee monotonicity and proportionality for solid coalitions.

\subsection{Committee Monotonicity}

The idea of committee monotonicity is that if some candidate is selected for a committee size $k$, then it should also be selected for every committee size $k'>k$. The formal definition is as follows.

\begin{definition}[Committee monotonicity]
	A committee voting rule $f$ is \emph{committee monotone} if $f(R,k) \subseteq f(R,{k+1})$ for all preference profiles $R$ and committee sizes $k \in \{1,\dots, m-1\}$.
\end{definition} 

\citet[Section 2.3.1]{FaliszewskiTrends2017} deem committee monotonicity a necessity for ex\-cell\-ence-based elections, where the goal is to choose the individually best candidates for the considered problem. \citet[Section 5]{Elkind2017} give the example of selecting finalists of a competition where they believe that committee monotonicity is ``imperative''. On the other hand, they write that ``in the context of proportional representation insisting on a committee-monotone rule may prevent us from selecting a truly representative committee''. They illustrate this with an example (attributed to \citet{Black1958}) of a single-peaked society, where for $k = 1$ it is most natural to select the median candidate while for $k = 2$, the committee should intuitively consist of a ``moderate left-wing'' and a ``moderate right-wing'' candidate. While it is true that committee monotonicity comes at a cost for the representativeness of the chosen committees, we believe that it is a necessary requirement for many applications of proportionally representative voting (see the introduction for examples).

As we explain in \Cref{sec:rank-jr}, a committee monotone rule can be transformed into a rule that returns a ranking of the candidates instead of committees. This observation is also discussed by \citet[Theorem 2]{Elkind2017}.

\subsection{Proportionality for Solid Coalitions}\label{subsec:PSC}

Next, we introduce our central fairness concept called proportionality for solid coalitions (PSC), which is due to \citet{Dummett1984}. The idea of this axiom is that, if there is a sufficiently large set of voters $N'\subseteq N$ that all prefer the candidates in a subset $C'\subseteq C$ to the candidates in $C\setminus C'$, then this group should be represented by a number of candidates in $C'$ that is proportional to the size of~$N'$. We will first define this axiom for strict preferences before presenting a variant called inclusion \psc (\ipsc) due to \citet{azizlee2021} that can also accommodate weak preferences. Note that we will define these axioms as properties of committees; a committee voting rule satisfies PSC or IPSC if its selected committee always satisfies the given axiom. 

To formalize PSC, we define a \emph{solid coalition} for a set of candidates $C'\subseteq C$ as a group of voters $N'\subseteq N$ such that $c'\succ_i c$ for all voters $i\in N'$ and candidates $c'\in C'$, $c\in C\setminus C'$. In this case, we also say that the voters in $N'$ \emph{support} the candidates $C'$. We emphasize that the voters in $N'$ do not have to agree on the order of the candidates in $C'$ and that a voter can be part of multiple solid coalitions. 
Now, proportionality for solid coalitions postulates that each solid coalition $N'$ of size $|N'|>\ell\cdot \frac{n}{k+1}$ (for some $\ell\in\mathbb N$) is represented by at least $\ell$ candidates or the set $C'$ if $|C'|<\ell$.

\begin{definition}[Proportionality for solid coalitions \citep{Dummett1984}]
	\label{def:psc}
	A committee $W$ satisfies \emph{proportionality for solid coalitions} (\psc or Droop-\psc) for a preference profile $R$ and a committee size~$k$ if for all integers $\ell\in\mathbb{N}$ and solid coalitions $N'$ supporting a set $C'$ with $|N'| > \ell \cdot \frac{n}{k+1}$, it holds that $C'\subseteq W$ or $|W\cap C'|\geq \ell$.
\end{definition}

Some authors consider the \psc axiom together with the \emph{Hare quota} of $\frac{n}{k}$ instead of the Droop quota of $\frac{n}{k+1}$. This gives rise to a weaker axiom that we call \emph{Hare-PSC}, for which we replace the condition $|N'| > \ell \cdot \frac{n}{k+1}$ by $|N'| \ge\ell \cdot \frac{n}{k}$ in \Cref{def:psc}. All our positive results will apply to the stronger Droop-PSC, while all our counterexamples will work even for Hare-PSC.

\iflatexml
\begin{figure}[t]
	\centering
	\includegraphics{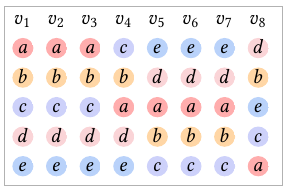}
	\caption{An example for PSC.}
	\label{fig:psc_exp}
\end{figure}
\else
\begin{wrapstuff}[r,type=figure,width=5cm]
	\centering
	\begin{tikzpicture}
		\node at (0,0) [draw=black!30] (step1) {	
			\votermultiplicity{$v_1$}{\weakorder{{a},{b},{c},{d},{e}}}
			\votermultiplicity{$v_2$}{\weakorder{{a},{b},{c},{d},{e}}}
			\votermultiplicity{$v_3$}{\weakorder{{a},{b},{c},{d},{e}}}
			\votermultiplicity{$v_4$}{\weakorder{{c},{b},{a},{d},{e}}}
			\votermultiplicity{$v_5$}{\weakorder{{e},{d},{a},{b}, {c}}}
			\votermultiplicity{$v_6$}{\weakorder{{e},{d},{a},{b}, {c}}}
			\votermultiplicity{$v_7$}{\weakorder{{e},{d},{a},{b}, {c}}}
			\votermultiplicity{$v_8$}{\weakorder{{d},{b},{e},{c}, {a}}}
		};
	\end{tikzpicture}
	\vspace{-7pt}
	\caption{An example for PSC.}
	\label{fig:psc_exp}
\end{wrapstuff}
\fi
\begin{example}[Droop-PSC and Hare-PSC]
	Consider the profile shown in \Cref{fig:psc_exp} with $n=8$ voters for $k = 2$. Here, the Hare quota $\frac{n}{k}$ is $4$ while the Droop quota $\frac{n}{k+1}$ is $\frac{8}{3} < 3$. For this instance, Hare-PSC imposes only one non-trivial inclusion: the voters $\{v_1,\dots, v_4\}$ for a solid coalition for $\{a,b,c\}$, so one of these candidates must be selected. The candidate $d$ does not need to be included, as for instance, voters $\{v_5,\dots, v_8\}$ do not have any common non-trivial prefix containing $d$. Droop-PSC, on the other hand, forces the selected committee to be $\{a,e\}$, as three voters are enough to force a single candidate to be part of the committee. 
\end{example}

We will next present a generalization of \psc to weak preferences due to \citet{azizlee2021}. To this end, we call a set of voters $N'\subseteq N$ a \emph{generalized solid coalition} supporting a set of candidates $C'\subseteq C$ if for all voters $i\in N'$, we have $c' \succsim_i c$ for all candidates $c'\in C'$ and $c\in C\setminus C'$. That is, generalized solid coalitions only have to weakly prefer the candidates in $C'$ to those in $C\setminus C'$. 
Furthermore, following \citet{azizlee2021}, we define the \emph{periphery} 
\[
	\periphery{C'}{N'} = \{ c \in C : \text{ there exists } {i \in N'} \allowbreak \text{and } c' \in C' \text{ such that } c \succsim_i c' \}
\]
of $C'$ with respect to $N'$. This is a ``closure'' of $C'$, containing $C'$ as well as all candidates that at least one member of $N'$ weakly prefers to a member of $C'$. To construct the periphery, we can imagine going through the indifference classes of the weak order for each member of a generalized solid coalition $N'$. The most-preferred indifference classes will consist only of members of $C'$. These are potentially followed by a single ``mixed'' indifference class containing both members of $C'$ and other candidates. All these other candidates are placed in the periphery. The remaining indifference classes must be disjoint from $C'$.

\iflatexml
\begin{figure}[t]
	\centering
	\includegraphics{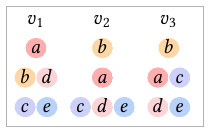}
	\caption{A generalized solid coalition.}
	\label{fig:periphery}
\end{figure}
\else
\begin{wrapstuff}[r,type=figure,width=4.6cm]
	\vspace{-0.1cm}
	\begin{tikzpicture}
		\node at (0,0) [draw=black!30] (step1) {
			\votermultiplicity{$v_1$}{\weakorder{{a},{b, d}, {c, e}}}
			\votermultiplicity{$v_2$}{\weakorder{{b},{a},{c,d, e}}}
			\votermultiplicity{$v_3$}{\weakorder{{b},{a,c},{d,e}}}
		};
	\end{tikzpicture}
	\vspace{-7pt}
	\caption{A generalized solid coalition.}
	\label{fig:periphery}
\end{wrapstuff}
\fi
\begin{example}[Periphery]
	The three voters with the weak preferences shown in \Cref{fig:periphery} form a generalized solid coalition $N' = \{v_1, \dots, v_3\}$ supporting $C' = \{a,b\}$. The periphery of $C'$ with respect to $N'$ is $\periphery{C'}{N'} = \{a,b,c,d\}$. The periphery contains $c$ due to voter $v_3$ and $d$ due to voter $v_1$. The smaller generalized solid coalition $N' = \{v_1, v_2\}$ also supports $C' = \{a,b\}$, but for this coalition, we have $\periphery{C'}{N'} = \{a,b,d\}$.
\end{example}

Whenever a candidate from the periphery is included in the committee, we will count the candidate as contributing to the satisfaction of $N'$.
Thus, inclusion \psc (\ipsc{}) demands that, for every generalized solid coalition $N'$ supporting a set $C'$,  a number of candidates proportional to~$\rvert N'\rvert$ needs to the chosen from the periphery $\periphery{C'}{N'}$. 

\begin{definition}[Inclusion \psc \citep{azizlee2021}]
A committee $W$ satisfies \emph{inclusion \psc (\ipsc)} for a profile $R$ and committee size $k$ if for all integers $\ell\in\mathbb{N}$ and generalized solid coalitions $N'$ supporting a set $C'$ with $|N'|> \ell \cdot \frac{n}{k+1}$, it holds that $C' \subseteq W \text{ or } |W\cap \periphery{C'}{N'}|\ge \ell$.
\end{definition}

Again, we can define a version of IPSC for the Hare quota by only allowing coalitions with $|N'|\geq \ell\cdot \frac{n}{k}$. We call the resulting property Hare-IPSC and note that it is weaker than (Droop-)\ipsc.

\section{\psc and Committee Monotonicity}\label{sec:psc}

\iflatexml
\begin{figure}[t]
	\centering
	\includegraphics{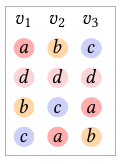}
	\caption{A difficulty with satisfying \psc and committee monotonicity simultaneously.}
	\label{fig:psc_committee_monotonicity_conflict}
\end{figure}
\else
\begin{wrapstuff}[r,type=figure,width=3.3cm]
	\centering
	\begin{tikzpicture}
		\node at (0,0) [draw=black!30] (step1) {	
			\votermultiplicity{$v_1$}{\weakorder{{a},{d},{b},{c}}}
			\votermultiplicity{$v_2$}{\weakorder{{b},{d},{c},{a}}}
			\votermultiplicity{$v_3$}{\weakorder{{c},{d},{a},{b}}}
		};
	\end{tikzpicture}
	\vspace{-7pt}
	\caption{A difficulty with satisfying \psc and committee monotonicity simultaneously.}
	\label{fig:psc_committee_monotonicity_conflict}
\end{wrapstuff}
\fi
\noindent
We will now investigate the compatibility of \psc and committee monotonicity. The combination of these axioms can be tricky to satisfy. Consider the example shown in \Cref{fig:psc_committee_monotonicity_conflict}.
For $k = 1$, the \psc axiom imposes no constraints, so all of $\{a\}$, $\{b\}$, $\{c\}$, and $\{d\}$ are committees satisfying \psc. However, $\{d\}$ seems to be the best committee, since it is a compromise candidate (and indeed $d$ is the Condorcet winner and also the Borda winner). However, for $k = 3$, the unique \psc committee is $\{a,b,c\}$. Therefore, a committee monotone rule satisfying \psc may not choose $\{d\}$ for $k = 1$. However, several rules \emph{do} choose $d$, including the Expanding Approvals Rule, the QBS rule, the Schulze STV rule \citep[Section 8]{Schu18a}, and CPO-STV \citep{Tideman95,TidemanBetterVoting2000}.\footnotemark\ 
STV also fails committee monotonicity; an example is given by \citet[Proposition 3]{Elkind2017}.
In \Cref{app:existing-rules-and-comm-mon}, we construct additional counterexamples showing that all rules in Dummett's family of \emph{Quota Preference Score} rules also fail committee monotonicity. 

\footnotetext{For Schulze STV, no rigorous proof that it satisfies PSC has been published. For CPO-STV, it is open whether it satisfies PSC \citep[page 282]{tideman2017collective}.}

Since none of the known committee voting rules satisfy both \psc and committee monotonicity, we will design new voting rules to achieve both axioms simultaneously. As our first result we show that for strict preferences, there is a simple way to achieve PSC and committee monotonicity by modifying existing rules. To this end, we introduce the \emph{reverse sequential rule} $f^\mathit{RS}$ of a committee voting rule $f$. Roughly, these reverse sequential rules compute the winning committee by repeatedly using the original rule $f$ to identify alternatives that are to be removed from the winning committee. To make this more formal, we let $R|_{X}$ denote the restriction of a preference profile $R$ to the set $X \subseteq C$, i.e., we derive $R|_X$ by deleting the alternatives $C\setminus X$ from $R$. Then, the reverse sequential rule $f^\mathit{RS}$ of a committee voting rule $f$ is defined recursively by $f^\mathit{RS}(R,m)=C$ and for all $k\in \{m-1, \dots, 1\}$, $f^\mathit{RS}(R, k)=f(R|_{X}, k)$ where $X = f^\mathit{RS}(R,k+1)$.

\begin{example}[Reverse sequential rules]
	Let us reconsider the profile in \Cref{fig:psc_committee_monotonicity_conflict}, where most PSC methods fail committee monotonicity because they select the committee $\{d\}$ for $k = 1$. However, if we run these rules in the reverse sequential mode, we first compute them for $k = 3$, when they return the unique PSC committee $\{a,b,c\}$. Thereafter, we consider the profile with $d$ removed. Thus, such reverse sequential rules will not select $\{d\}$ for $k = 1$.
\end{example}

We show next that for strict preferences, the reverse sequential rule $f^\mathit{RS}$ satisfies committee monotonicity and PSC if the original rule $f$ satisfies PSC. This means that, e.g., the reverse sequential rule of STV satisfies both of our desiderata. 

\begin{theorem}\label{thm:revseq}
	Under strict preferences, for every committee voting rule $f$ that satisfies \psc, the reverse sequential rule $f^\mathit{RS}$ satisfies \psc and is committee monotone. 
\end{theorem}
\begin{proof}
	The reverse sequential rule $f^\mathit{RS}$ of a committee voting rule $f$ satisfies committee monotonicity because $f^\mathit{RS}(R, k)=f(R|_{f^\mathit{RS}(R,k+1)}, k)\subseteq f^\mathit{RS}(R,k+1)$ for all $k\in \{1,\dots, m-1\}$ and $R\in\mathcal{L}^N$. 
	
	Next, suppose that $f$ satisfies \psc. Fix a profile $R\in\mathcal{L}^N$, and write $W_k = f^\mathit{RS}(R,k)$ for all $k\in \{1,\dots, m\}$. Further, we fix some $k$, and we will show that $W_k$ satisfies \psc.
	For this, let $N' \subseteq N$ be a solid coalition such that $|N'|> \ell \frac{n}{k+1}$ for some $\ell\in\mathbb{N}$, and let $C' \subseteq C$ denote the set of candidates supported by $N'$. 
	Now, to establish \psc, we need to show that $|W_k \cap C'| \ge \min(|C'|, \ell)$. We prove by a backwards induction that $|W_t \cap C'| \ge \min(|C'|, \ell)$ for all $t \in \{m, m-1, \dots, k\}$.
	
	Clearly, at the start of the reverse sequential rule, when $t = m$, we have $|W_m \cap C'| = |C'| \ge \min(|C'|, \ell)$ because $W_m = C$. Suppose that we have shown that $|W_t \cap C'| \ge \min(|C'|, \ell)$ for some $t \in \{m, \dots, k+1\}$. If it even holds that $|W_t \cap C'| > \min(|C'|, \ell)$, then we are done, because $W_{t-1}$ is obtained from $W_t$ by deleting only one alternative, and thus $|W_{t-1} \cap C'| \ge \min(|C'|, \ell)$. So suppose that $|W_{t} \cap C'| = \min(|C'|, \ell)$. Note that $N'$ is a solid coalition with $|N'| > \ell \frac{n}{k+1} \ge \ell \frac{n}{t}$ supporting the set $W_{t}\cap C'$ in the profile $R|_{W_{t}}$. Since $f$ satisfies PSC, at least $\min(|W_{t}\cap C'|, \ell)$ candidates from $W_t\cap C'$ are chosen by $f$ for $R|_{W_t}$ and the committee size $t-1$.
	Because $|W_t\cap C'|=\min(|C'|,\ell)$, it holds that $\min(|W_t\cap C'|, \ell)=\min(|C'|,\ell)$ and thus $W_{t}\cap C' \subseteq f(R|_{W_{t}}, t-1) = W_{t - 1}$. It follows that $|W_{t-1}\cap C'| \ge \min(|C'|, \ell)$, establishing the induction step.
\end{proof}
	\iflatexml
	\begin{figure}[t]
		\centering
		\includegraphics{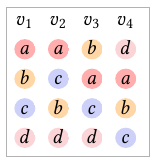}
		\caption{Not all committee monotone rules that satisfy \psc are reverse sequential rules.}
		\label{fig:rev_seq_not_everything}
	\end{figure}
	\else
	\begin{wrapstuff}[r,type=figure,width=2.6cm]
		\begin{tikzpicture}
			\node at (0,0) [draw=black!30] (step1) {
				\votermultiplicity{$v_1$}{\weakorder{{a},{b},{c},{d}}}
				\votermultiplicity{$v_2$}{\weakorder{{a},{c},{b},{d}}}
				\votermultiplicity{$v_3$}{\weakorder{{b},{a},{c},{d}}}
				\votermultiplicity{$v_4$}{\weakorder{{d},{a},{b},{c}}}
			};
		\end{tikzpicture}
		\vspace{-7pt}
		\caption{Not all committee monotone rules that satisfy \psc are reverse sequential rules.}
		\label{fig:rev_seq_not_everything}
	\end{wrapstuff}
	\fi
\Cref{thm:revseq} does not provide a characterization of committee monotone \psc rules; there exist such rules that cannot be expressed through the reverse sequential scheme. For example, on the profile in \Cref{fig:rev_seq_not_everything}, a committee monotone \psc rule could select the committees $\{c\}$, $\{c,a\}$, $\{c,a,b\}$, and $\{c,a,b,d\}$ for $k = 1,2,3,4$, respectively. However, after restricting the profile to the candidates in $\{c,a\}$, any PSC rule must select $\{a\}$. Thus, this sequence of committees cannot be obtained by a rule defined in a reverse sequential manner.

On a high level, the proof of \Cref{thm:revseq} considers a solid coalition which has a demand of the form ``at least $z=\min(|C'|, \ell)$ candidates out of this set of $t=|C'|$ candidates should be elected''. The proof then ``waits'' until all but exactly $z$ of those $t$ candidates have been eliminated, and then argues that these remaining $z$ candidates must all be included in the final committee.
Interestingly, this is the same strategy that appears in the proof that STV satisfies PSC [\citealp{Tideman95}, Footnotes 1 and 2; \citealp{Aziz2019}, Appendix C; \citealp{DePe24a}, Theorem 5.4].

As a consequence of this proof strategy, \Cref{thm:revseq} holds even for voting rules $f$ that only satisfy \emph{weak PSC}, a property introduced by \citet[Theorem 16.1]{Jans16a}. This property is defined just like in \Cref{def:psc}, but it only applies to sets $C'$ with $|C'| = \ell$.
\citet{Jans16a} discussed this property in the context of the Ordered Phragmén rule which fails PSC and only satisfies weak PSC. In exchange, the Ordered Phragmén rule is committee monotone and it satisfies \emph{candidate monotonicity}, for which the compatibility with PSC is open. The Ordered Phragm\'en can be defined as follows: every agent ``eats'' their highest-ranked candidate at unit speed until $k$ candidates are selected. When agents have fully eaten a candidate, it is added to the committee and the voters move on to their next-most-preferred candidate.
Since \Cref{thm:revseq} works for weak PSC, it implies that the reverse sequential version of the Ordered Phragmén rule satisfies PSC. However, it does not inherit the property of candidate monotonicity (see \Cref{prop:rev-seq-phragmen-fails-candidate-monotonicity} in the appendix).

	\iflatexml
		\begin{figure}[t]
			\centering
			\includegraphics{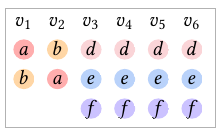}
			\caption{The Ordered Phragmén rule fails PSC.}
			\label{fig:weak_phragmen}
		\end{figure}
		\else
		\begin{wrapstuff}[r,type=figure,width=3.9cm]
			\begin{tikzpicture}
				\node at (0,0) [draw=black!30] (step1) {
					\votermultiplicity{$v_1$}{\weakorder{{a},{b},{}}}
					\votermultiplicity{$v_2$}{\weakorder{{b},{a},{}}}
					\votermultiplicity{$v_3$}{\weakorder{{d},{e},{f}}}
					\votermultiplicity{$v_4$}{\weakorder{{d},{e},{f}}}
					\votermultiplicity{$v_5$}{\weakorder{{d},{e},{f}}}
					\votermultiplicity{$v_6$}{\weakorder{{d},{e},{f}}}
				};
			\end{tikzpicture}
			\vspace{-7pt}
			\caption{The Ordered Phragmén rule fails PSC.}
			\label{fig:weak_phragmen}
		\end{wrapstuff}
	\fi
\begin{example}[Reverse Sequential Ordered Phragmén]
	Consider the profile in \Cref{fig:weak_phragmen}, where the missing candidates of each ranking can be inserted arbitrarily below the given prefix. For this profile, the Ordered Phragmén rule fails PSC (even Hare-PSC). Indeed, for $k = 3$, it elects $\{d, e, f\}$ because voters 3 to 6 jointly eat a candidate every 0.25 time units, so $\{d,e,f\}$ are completely eaten within 0.75 time units, while $a$ and $b$ take 1 time unit to eat. However, $N' = \{v_1, v_2\}$ form a solid coalition supporting $C' = \{a,b\}$ while $|N'| = 2 \ge \frac{n}{k}$. On the other hand, running the rule in reverse sequential mode, one of candidates $a$ and $b$ will be eliminated for $k = 4$ (let's say $b$) and then the committee for $k = 3$ will be $\{a,d,e\}$ which satisfies PSC. 
	This example illustrates an observation of \citet[Remark 16.6]{Jans16a} that a reason that the Ordered Phragmén rule fails PSC is its lack of a mechanism to eliminate candidates.
\end{example}

\iflatexml
\begin{figure}[t]
	\centering
	\includegraphics{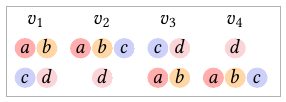}
	\caption{For weak preferences, reverse sequential rules may fail IPSC.}
	\label{fig:rev_seq_doesn't_work_for_weak}
\end{figure}
\else
\begin{wrapstuff}[r,type=figure,width=4.9cm]
	\begin{tikzpicture}
		\node at (0,0) [draw=black!30] (step1) {
			\votermultiplicity{$v_1$}{\weakorder{{a,b},{c,d}}}
			\votermultiplicity{$v_2$}{\weakorder{{a,b,c},{d}}}
			\votermultiplicity{$v_3$}{\weakorder{{c,d},{a,b}}}
			\votermultiplicity{$v_4$}{\weakorder{{d},{a,b,c}}}
		};
	\end{tikzpicture}
	\vspace{-7pt}
	\caption{For weak preferences, reverse sequential rules may fail IPSC.}
	\label{fig:rev_seq_doesn't_work_for_weak}
\end{wrapstuff}
\fi
Unfortunately, \Cref{thm:revseq} does not extend to weak preferences and \ipsc. For example, in the profile shown in \Cref{fig:rev_seq_doesn't_work_for_weak}, the committee $\{a,b,c\}$ satisfies IPSC, and the committee $\{a,b\}$ satisfies IPSC after restricting the profile to $\{a,b,c\}$. However, $\{a,b\}$ violates IPSC 
because the last two voters form a generalized solid coalition supporting $\{d\}$ while no candidate from the periphery $\{c,d\}$ is selected. For another example, consider the proportional approval voting (PAV) rule which is defined for approval preferences, a special case of weak preferences, and satisfies \ipsc \citep[Proposition 11]{azizlee2021}.
If \Cref{thm:revseq} held for weak preferences, then the reverse sequential version of PAV should also satisfy IPSC, and thus also the weaker axioms PJR and JR \citep[Proposition 10]{azizlee2021}.
But reverse sequential PAV fails JR \citep{Aziz17a}. 

\section{The Solid Coalition Refinement Rule}\label{subsec:SCR}

A main drawback of committee voting rules obtained by the reverse sequential transformation is that they are rather technical and hard to analyze. For example, we saw that the Ordered Phragmén rule loses its candidate monotonicity property when run in the reverse sequential mode. In addition, these rules can be computationally intensive, since we need to run $m - k$ rounds of the base rule to determine $k$ winners. This is a concern especially in rank aggregation applications, where users are usually most interested in seeing the top candidates, but reverse sequential rules will start out being busy to determine the bottom candidates.
These drawbacks motivate the search for a more direct rule that combines \psc and committee monotonicity. We thus present another rule called the \emph{Solid Coalition Refinement rule} (or SCR for short), which satisfies both of our desiderata. Moreover, this rule even generalizes to weak preferences and satisfies \ipsc in this case.

Since the SCR rule is inspired by the D'Hondt method of apportionment, we first give a description of this apportionment method. In apportionment, a set of $p$ parties with their respective vote counts $v_1, \dots, v_p$ need to be assigned a total of $h$ seats. The D'Hondt method of apportionment intuitively tries to assign seats such that \emph{each seat represents as many voters as possible}. For example, if a party with 100 votes is assigned 2 seats, each seat represents 50 voters. Now, the D'Hondt method proceeds sequentially, in each step assigning the next seat to a party that optimizes the per-seat representation. Formally, this means selecting a party $i$ that maximizes $v_i/(s_i + 1)$, where $s_i$ is the number of seats assigned to $i$ thus far. Note that $v_i/(s_i + 1)$ corresponds to the per-seat representation should the next seat go to party $i$.
This procedure is repeated until all $h$ seats are filled.

\begin{example}
	Suppose there are three parties with vote counts $v_1 = 100$, $v_2 = 60$, and $v_3 = 40$. The first seat is given to party 1, which received the most votes. The second seat is given to party 2 which represents 60 voters. It is not given to party 1, because that would lead to a per-seat representation of only 50 voters. However, the third seat does go to party 1, and the fourth seat goes to party 3. Thus, if $h = 4$, then the D'Hondt seat distribution would be $s_1 = 2$ and $s_2 = s_3 = 1$.
\end{example}

In the committee voting setting, there are no parties, only individual candidates. However, we can view sets of candidates as ``virtual parties'' that receive votes from the solid coalitions supporting them. Thus, each subset $C' \subseteq C$ of candidates forms a party, supported by $v_{C'} = |\{i \in N : C' \succ_i C \setminus C'\}|$ voters. The SCR rule will sequentially build up a committee $W$ based on this idea. At each step of the procedure, we view the party $C'$ as having been assigned $s_{C'} = |C' \cap W|$ seats thus far. (Since parties overlap, the same seat counts for several parties, but this is not a problem.) Then, following the D'Hondt philosophy, we can identify the party that most deserves being assigned the next seat: we pick a party $C'$ maximizing $v_{C'}/(s_{C'} + 1)$. In selecting this party, the method ignores parties $C'$ that have already been fully satisfied in the sense that $C' \subseteq W$.

Once we have identified a most-deserving party $C'$, we still need to decide which candidate from $C' \setminus W$ to add to the committee. If $|C' \setminus W| = 1$, it is clear that we add the sole candidate in $C'\setminus W$. On the other hand, if $|C'\setminus W|>1$, we choose the candidate based on the same philosophy of maximizing the representation of each seat. Thus, we will look for a strict sub-party of $C'$ (i.e., $C'' \subsetneq C'$) that maximizes $v_{C''}/(s_{C''} + 1)$ and that includes a non-selected candidate (i.e., $C'' \setminus W \neq \emptyset$). We may further repeat this refinement step until we find a party contained in $C'$ that includes just a single unselected candidate. This is the candidate that the SCR rule adds to the committee.
This procedure is repeated until $k$ seats are filled. 

Let us further generalize the SCR rule to weak preferences and \ipsc.
Compared to the strict order case, the concept of virtual parties is necessarily more complex in this setting.
In particular, a virtual party will now not just refer to a set of candidates, but to a pair $(N', C')$ where $N'$ is a generalized solid coalition supporting the set $C'$ of candidates.
The number of voters for this virtual party is $|N'|$, and the number of ``seats'' assigned to the virtual party by a committee $W$ is $|W\cap \periphery{C'}{N'}|$. 
Thus, let us define
\[
	\rho(W,N',C')=\frac{|N'|}{|W\cap \periphery{C'}{N'}|+1}
\]
as the \emph{underrepresentation value} of $(N',C')$ under the committee $W$. 
Hence, a large underrepresentation value means that the generalized solid coalition $N'$ is far from being proportionally represented.
Indeed, we can establish a direct relationship between the underrepresentation value and the \ipsc property.

\begin{proposition} \label{prop:underrep}
	A committee $W$ satisfies \ipsc for a profile $R$ and a committee size $k$ if and only if $\rho(W, N', C')\leq \frac{n}{k+1}$ for every generalized solid coalition $N'$ that supports a set $C'$ with $C'\not\subseteq W$. 
\end{proposition}
\begin{proof}
	Fix a profile $R$ and a committee size $k$. First, suppose that the committee $W$ satisfies \ipsc for $R$ and $k$. Let $N'$ be a generalized solid coalition supporting a candidate set $C'$ with $C'\not\subseteq W$ and let $\ell^*$ denote the maximal integer such that $|N'|>\ell^* \frac{n}{k+1}$, which means that $|N'|\leq (\ell^*+1)\frac{n}{k+1}$. By \ipsc, we have that $|W\cap \periphery{C'}{N'}|\geq \ell^*$, so
	\[\rho(W,N',C')=\frac{|N'|}{|W\cap \periphery{C'}{N'}|+1}\leq \frac{|N'|}{\ell^*+1}\leq \frac{n}{k+1}.
	\]
	
	Conversely, suppose that the committee $W$ fails \ipsc for $R$ and $k$. Thus, there is a generalized solid coalition $N'$ supporting a set $C'$ and an integer $\ell\in \mathbb{N}$ such that $|N'|>\frac{\ell n}{k+1}$, $C'\not\subseteq W$, and $|W\cap \periphery{C'}{N'}|< \ell$. The underrepresentation value for $(N',C')$ is 
	\[\rho(W,N',C')=\frac{|N'|}{|W\cap \periphery{C'}{N'}|+1}\geq  \frac{|N'|}{\ell}> \frac{n}{k+1},
	\]
	which proves the direction from right to left of our proposition.
\end{proof}

As a generalization of the philosophy of the D'Hondt method, the generalized solid coalition with the highest underrepresentation value has the strongest claim to decide the next seat.
As a final piece of notation, we denote by $\Phi(R,W,D)$ the set of pairs $(N', C')$ where $N'$ is a generalized solid coalition supporting the set $C'$ such that $C' \not\subseteq W$ (the coalition is not fully satisfied) and $C' \subsetneq D$ (it refines a given set $D$).

We are now ready to define the SCR rule. Starting from $W=\emptyset$, this rule computes the winning committee by repeatedly adding single candidates to $W$. The next candidate is chosen as follows: we first identify the generalized solid coalition $N'$ supporting a set $C'\not\subseteq W$ such that $(N',C')$ maximizes the underrepresentation value $\rho(W, N', C')$ among all generalized solid coalitions in $\Phi(R,W,C)$. The goal is then to select a candidate from $C'\setminus W$. Analogously to the case of strict preferences, to decide which candidate from $C'\setminus W$ to choose, we identify the generalized solid coalition $N''$ supporting a candidate set $C''$ with $C''\subsetneq C'$ and $C''\not\subseteq W$ such that $\rho(W,N'',C'')$ is maximal among all generalized solid coalitions in $\Phi(R,W,C')$. By repeating this step, we will eventually arrive at a generalized solid coalition $N^*$ supporting a set $C^*$ such that $|C^*\setminus W|=1$, and we add the single candidate in $C^*\setminus W$ to $W$. Put differently, starting from $D=C$, we repeatedly update $D$ to be the set of candidates $C^*$ corresponding to a generalized solid coalition $N^*$
such that $(N^*, C^*)$ maximizes $\rho(W, N^*,C^*)$ among all elements in $\Phi(R,W,D)$ until $|D\setminus W|=1$. Then, we add the candidate in $D\setminus W$ to $W$. A pseudocode description of the SCR rule is given in \Cref{alg:SCR}.

We note that multiple generalized solid coalitions in $\Phi(R,W,D)$ may have the same maximal underrepresentation value in some steps; in such cases, we assume that ties are broken by an arbitrary but fixed ranking $\rhd$ over the sets of candidates $C'\subseteq C$. That is, if two generalized solid coalitions $N'$ and $N''$ with candidate sets $C'$ and $C''$ have the same maximal underrepresentation value in some step of the SCR rule, we choose $(N', C')$  if $C'\rhd C''$ and $(N'',C'')$ otherwise.

\begin{algorithm}[t]
	\caption{The Solid Coalition Refinement Rule}
	\label{alg:SCR}
	\SetKwInOut{Input}{Input}
	\SetKwInOut{Output}{Output}
	\SetKwComment{Comment}{$\triangleright$\ }{}
	\SetInd{0.2em}{0.7em} 
	\Input{A preference profile $R$ and committee size $k$}
	\Output{A committee of $k$ candidates}
	$W \gets \emptyset$
	
	\For{$i\in \{1,\dots, k\}$\label{line:outerloop}} { 
		$D \gets C$
		
		\While{$|D\setminus W|>1$\label{line:whileloop}}  {
			${\Phi(R,W,D) \!\gets\! \{(N',C')\! :\! \text{$N'$ is a generalized solid coalition supporting $C'\!\subsetneq\! D$ with $C'\!\not\subseteq\!W$}\}}$\label{line:define-phi}

			$(N^*,C^*)\gets \argmax_{(N',C')\in \Phi(R,W,D)} \frac{|N'|}{|W\,\cap\, \periphery{C'}{N'}|\,+\,1}$\label{line:choose-N*-and-C*}
			
			$D \gets C^*$
		}
		$W\gets W\cup D$\label{line:add-candidate}
	}
	\Return $W$
\end{algorithm}

\begin{example}\label{example:alg}
	Consider the preference profile in \Cref{fig:SCRexample} with $n = 5$ voters and $m = 4$ candidates. 
	With a committee size $k = 3$, the SCR rule runs as follows.

	\begin{enumerate}[label=\arabic*), leftmargin=*, itemsep=3pt]
		\item %
		\iflatexml
		\begin{figure}[t]
			\centering
			\includegraphics{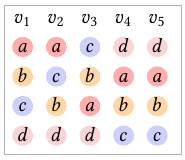}
			\caption{An example for the SCR rule.}
			\label{fig:SCRexample}
		\end{figure}
		\else
		\begin{wrapstuff}[r,type=figure,width=3.4cm]
			\begin{tikzpicture}
				\node at (0,0) [draw=black!30] (step1) {
					\votermultiplicity{$v_1$}{\weakorder{{a},{b},{c},{d}}}
					\votermultiplicity{$v_2$}{\weakorder{{a},{c},{b},{d}}}
					\votermultiplicity{$v_3$}{\weakorder{{c},{b},{a},{d}}}
					\votermultiplicity{$v_4$}{\weakorder{{d},{a},{b},{c}}}
					\votermultiplicity{$v_5$}{\weakorder{{d},{a},{b},{c}}}
				};
			\end{tikzpicture}
			\vspace{-7pt}
			\caption{An example for the SCR rule.}
			\label{fig:SCRexample}
		\end{wrapstuff}
		\fi
		
		Initially, $D = \{a,b,c,d\}$ and $W=\emptyset$. Then:
		
			\begin{itemize}[leftmargin=*, topsep=1pt]
				\item The solid coalition $N' = \{v_1, v_2, v_3\}$ supporting $C' = \{a,b,c\}$ has $\rho(W, N', C')=3$, which is maximal among all solid coalitions in $\Phi(R,W,D)$. SCR sets $D = \{a,b,c\}$.
				\item The solid coalition $N' = \{v_1,v_2\}$ supporting $C' = \{a\}$ has $\rho(W, N', C')=2$, which is maximal among all solid coalitions in $\Phi(R,W,D)$. SCR sets $D = \{a\}$.
			\end{itemize}
	
		Candidate $a$ is selected, so $W=\{a\}$. 
		\item $D$ is reset to  $\{a,b,c,d\}$. Then:
		
		\begin{itemize}[leftmargin=*, topsep=1pt]
		\item The solid coalition $N' = \{v_4,v_5\}$ supporting $C' = \{d\}$ has $\rho(W, N', C')=2$, which is maximal among all solid coalitions in $\Phi(R,W,D)$. SCR sets $D = \{d\}$.
		\end{itemize}
		
		 Candidate $d$ is selected, so $W=\{a,d\}$. 
		\item $D$ is reset to  $\{a,b,c,d\}$. Then:
		
		\begin{itemize}[leftmargin=*, topsep=1pt]
		\item The solid coalition $N' = \{v_1, v_2, v_3\}$ supporting $C' = \{a,b,c\}$ has $\rho(W, N', C')=\frac{3}{2}$, which is maximal among all solid coalitions in $\Phi(R,W,D)$. SCR sets $D = \{a,b,c\}$.
		
		\item The solid coalition $N' = \{v_3\}$ supporting $C' = \{c\}$ has $\rho(W, N', C')=1$, which is maximal among all solid coalitions in $\Phi(R,W,D)$. SCR sets $D = \{c\}$.
		\end{itemize}
		
	Candidate $c$ is selected, so the final committee is $W=\{a,d,c\}$.
	\end{enumerate} 
\end{example}

We next show that the SCR rule is well-defined and runs in polynomial time if the voters have strict preferences. Due to space restrictions, we defer the full proof of this claim to \Cref{app:scrterm}.

\begin{restatable}{proposition}{scrterm}
	\label{lem:scr-time}
	The SCR rule always terminates and produces a committee of the target size $k$.
	Furthermore, for strict preferences, it can be implemented to run in polynomial time.
\end{restatable}
\begin{proof}[Proof Sketch]
	For showing that the SCR rule is well-defined, we note that each iteration of the for-loop (\cref{line:outerloop}) adds exactly one candidate to $W$. This is true because it holds for every pair $(N', C')$ in $\Phi(R,W,D)$ that $|C'\setminus W|\geq 1$ and the while-loop (\cref{line:whileloop}) is only exited when $|D\setminus W|\leq 1$. Moreover, it can be shown that the set $\Phi(R,W,D)$ is always non-empty during the execution of SCR, so the rule indeed produces a committee of the desired committee size $k$. Next, the SCR rule runs for strict preferences in polynomial time because, in this case, the set of (voter-maximal) solid coalitions can be efficiently computed and contains at most $mn$ elements. We hence can solve the maximization problem in \cref{line:choose-N*-and-C*} by iterating through all voter-maximal solid coalitions.
\end{proof}

A natural follow-up question to \Cref{lem:scr-time} is whether the SCR rule can also be computed in polynomial time for weak preferences.
Unfortunately, there is no clear way to compute \cref{line:choose-N*-and-C*} in this case. Indeed, if we could solve that maximization problem for every profile $R$ and committee $W$, we could also decide whether there is a generalized solid coalition $N'$ supporting a set $C'$ such that $C'\not\subseteq W$ and $\rho(W,N',C')>\frac{n}{k+1}$. However, this is equivalent to deciding whether the committee $W$ satisfies \ipsc (\Cref{prop:underrep}), which is a coNP-complete problem \citep{BrPe23a}.

Finally, we will show that the SCR rule satisfies committee monotonicity and \ipsc.
\begin{theorem}	\label{thm:SCR-ipsc-and-monotone}
	Even for weak preferences,
	the SCR rule is committee monotone and satisfies \ipsc.
\end{theorem}
\begin{proof}
	Since the SCR rule selects candidates sequentially and independently of the target committee size, it follows immediately that it satisfies committee monotonicity. 
	
	Next, to show that SCR satisfies \ipsc, we fix a profile $R$ and a committee size $k$. Moreover, for all $t\in \{1,\dots, k\}$, we let  $c_t$ denote the $t$-th candidate that SCR adds to the winning committee for $R$, and we define $W^t=\{c_1,\dots, c_t\}$ for all $t\in \{1,\dots,k\}$ and $W^0=\emptyset$. 
	For our proof, we assume that each voter $i \in N$ has a virtual budget $b_i = 1$ and that the candidates will be bought using these budgets for a price of $\frac{n}{k+1}$. We will next show that, for every $t\in \{1,\dots, k\}$, there is a payment scheme for $W^t$ such that 
	\begin{enumerate}
		\item[(i)] the budgets of all voters are always non-negative, and
		\item[(ii)] if a voter $i$ is part of a generalized solid coalition $N'$ supporting a set $C'$ such that $C\not\subseteq W^t$ and $\rho(W^t, N', C')>\frac{n}{k+1}$, then voter $i$ only spends a non-zero amount on candidates in $\periphery{C'}{N'}$.
	\end{enumerate}
	To derive such a scheme, we fix $t\in \{1,\dots, k\}$ and inductively assume that there is a payment scheme for the committee $W^{t-1}$ that satisfies conditions (i)  and (ii). 
	If $t=1$, such a scheme exist for $W^{t-1}$ as $W^0=\emptyset$ and no money has been spent. 
	We next explain how to extend the payment scheme for $W^{t-1}$ to~$W^t$ by reasoning how to pay for the candidate $c_t$. 
	
	For this, we proceed with a case distinction and first assume that 
	$\rho(W^{t-1}, N', C')\leq\frac{n}{k+1}$ 
	for all generalized solid coalitions $N'$ that support a set $C'$ with $C'\not\subseteq W^{t-1}$. In this case, we deduct the money for candidate $c_t$ arbitrarily from the budgets of the voters while ensuring that no budget becomes negative. This is possible as the voters' total budget exceeds the necessary budget to pay for $k$ candidates, i.e., $n>k\frac{n}{k+1}$. Moreover, it holds that 
	$\rho(W^t, N', C')\leq \rho(W^{t-1}, N', C')\leq \frac{n}{k+1}$ for all generalized solid coalitions $N'$ with set~$C'\not\subseteq W^t$, so condition (ii) holds trivially.
	
	For the second case, suppose that there is a generalized solid coalition $N'$ supporting a set $C'$ such that $C'\not\subseteq W^{t-1}$ and  
	$\rho(W^{t-1}, N', C')>\frac{n}{k+1}$.
	In this case, let $N^*$ and $C^*$ denote the last generalized solid coalition and the corresponding set of candidates in the execution of the while-loop of SCR (\cref{line:whileloop}) that satisfies these conditions. By condition (ii), the voters in $N^*$ have only spent money on the candidates in $\periphery{C^*}{N^*}$ so far. Hence, these voter have spent a total budget of at most $|W^{t-1} \cap \periphery{C^*}{N^*}|\frac{n}{k+1}$. By rearranging the assumption that $\rho(W^{t-1}, N^*, C^*)=\frac{|N^*|}{|W^{t-1} \cap \periphery{C^*}{N^*}|+1}>\frac{n}{k+1}$, we infer that $|N^*|>(|W^{t-1} \cap \periphery{C^*}{N^*}|+1)\frac{n}{k+1}$. 
	Therefore, these voters have a total remaining budget of at least $\frac{n}{k+1}$, so they can pay for the candidate $c_t$ without violating condition (i).
	
	It remains to show that this payment scheme for $W^t$ satisfies condition (ii). For this, let $N''$ denote a generalized solid coalition supporting a set $C''$ such that $C''\not\subseteq  W^t$ and $\rho(W^t, N'', C'')>\frac{n}{k+1}$. We assume for contradiction that there is a voter $i\in N''$ who spent money on candidates outside of $\periphery{C''}{N''}$. Since $\rho(W^{t-1}, N'', C'')\geq\rho(W^{t}, N'', C'')>\frac{n}{k+1}$, we infer from the induction hypothesis that voter $i$ has not spent money on candidates outside of $\periphery{C''}{N''}$ during the first $t-1$ steps. Hence, voter $i$ spent money on candidate $c_t$ and $c_t\not\in \periphery{C''}{N''}$. This means that $i\in N^*$, so $C^*$ forms a prefix of voter $i$'s preference relation. Further, because $c_t\not\in \periphery{C''}{N''}$, it follows that $C'' \subsetneq C^*$. However, this means that after SCR selected the solid coalition $N^*$ with set $C^*$ in the while-loop, there is another iteration of the while-loop such that
	\[
		\frac{|N'|}{|W^{t-1}\cap \periphery{C'}{N'}|+1}\geq \frac{|N'|}{|W^{t}\cap \periphery{C'}{N'}|+1}>\frac{n}{k+1}. 
	\]
	This violates the definition of $N^*$ and $C^*$, so the assumption that voter $i$ spent some money on $c_t$ is wrong and condition (ii) holds. 
	
	Finally, we will show that SCR satisfies \ipsc. Assume for contradiction that the committee $W^k$ fails \ipsc for $R$ and $k$. 
	Thus, there is a generalized solid coalition $N'$ supporting a set $C'$ and an integer $\ell\in\mathbb{N}$ such that $|N'|> \ell \frac{n}{k+1}$, $C'\not\subseteq W^k$, and $|W^k\cap \periphery{C'}{N'}|<\ell$. This implies that 
	\[
		\rho(W^k, N', C')=\frac{|N'|}{|W^k\cap \periphery{C'}{N'}|+1}\geq \frac{|N'|}{\ell}>\frac{n}{k+1}. 
	\]
	In turn, condition (ii) of the payment scheme shows that the voters in $N'$ only paid for candidates in $\periphery{C'}{N'}$. Thus, they spent a budget of at most $(\ell-1)\frac{n}{k+1}$. Since the total initial budget of these voters is $|N'|>\ell \frac{n}{k+1}$, their remaining budget is strictly larger than $\frac{n}{k+1}$. However, the total remaining budget of all voters after $k$ candidates have been bought is $n-\frac{kn}{k+1}=\frac{n}{k+1}$. Hence, there must be a voter with a negative budget, which violates condition (i). So, SCR satisfies \ipsc.
\end{proof}

\section{PSC and Irrelevant Voter Blocks}

We next use the SCR rule to answer an open question of \citet{GJM24a} by showing that there is a rule that satisfies both independence of losing voter blocs and \psc. 
In more detail, \citeauthor{GJM24a} study the setting of truncated (strict) preferences, i.e., voters have strict preferences but it is no longer necessary to rank all alternatives. 
More formally, a preference relation $\succsim_i$ is \emph{truncated} if it is a strict preference relation over a subset of the candidates. 
The notion of solid coalitions, and thus also \psc as well as the SCR rule, can be easily extended to truncated preferences. Specifically, solid coalitions are defined just as before and, in particular, only form for sets of voters who rank all alternatives in the supported set of candidates. Then, \psc and the SCR rule can be adapted to truncated preferences by using this new definition of solid coalitions.

Furthermore, \citet{GJM24a} suggest two consistency notions regarding the behavior of committee voting rules when some voters are removed from the election. One of these notions, independence of losing voter blocs, requires that the outcome should not change when we remove voters who only rank unchosen candidates. To formalize this notion, we define by $X({\succsim_i})$ the set of alternatives that are not ranked in the truncated preference relation $\succsim_i$. Moreover, $R_{-N'}$ denotes the profile derived from another profile $R$ by deleting the voters in $N'\subseteq N$ from $R$. Based on this notation, we now define independence of losing voter blocs.

\begin{definition}[Independence of losing voter blocs]
	A committee voting rule $f$ satisfies \emph{independence of losing voter blocs} if $f(R,k)=f(R_{-N'}, k)$ for all truncated preference profiles~$R$, committee sizes $k$, and sets of voters $N'\subsetneq N$ such that $f(R, k)\subseteq X({\succsim_i})$ for all $i\in N'$. 
\end{definition}

\citet{GJM24a} show that no classical proportional rules such as STV and EAR satisfy this property. The reason for this is, on a high-level, that these rules depend on the quota $\frac{n}{k+1}$, which changes if we remove voters. 
In fact, these authors even conjecture that no voting rule satisfies \psc and independence of losing voter blocs at the same time. We show that this conjecture is false as SCR (adapted to truncated rankings) satisfies both conditions.

\begin{theorem}\label{thm:truncated}
	For truncated preferences, the SCR rule satisfies \psc and independence of losing voter blocs.
\end{theorem}
\begin{proof}
	First, an analogous argument as in the proof of \Cref{thm:SCR-ipsc-and-monotone} shows that SCR satisfies \psc for truncated preferences. Hence, we focus on independence of losing voter blocs. To this end, we fix a profile $R$, a committee size $k$, and a set of voters $N'$ such that $W\subseteq X({\succsim_i})$ for all $i\in N'$ and the committee $W$ chosen by SCR for the profile $R$ and the committee size $k$. We need to show that SCR also chooses the committee $W$ for the profile $R_{-N'}$ and the committee size $k$. 
	For this, let $W'$ denote the intermediate committee and $D$ the current set of candidates of \Cref{alg:SCR} during some step of the execution of SCR for $R$ such that $|D\setminus W'|>1$. Moreover, let $(N^*,C^*)$ be the solid coalition and the set of candidates that is chosen next at \Cref{line:choose-N*-and-C*}. Since $W\subseteq X({\succsim_i})$ for all $i\in N'$, there must be for every voter $i\in N'$ a candidate $x\in C^*\cap X({\succsim_i})$. Otherwise, voter $i$ ranks the candidate that will be selected next because $C^*\subseteq C\setminus X({\succsim_i})$. Since solid coalitions form only over sets of voters $N''$ and sets $C''$ such that all voters $i\in N''$ rank all candidates in $C''$, this implies that $N^*\cap N'=\emptyset$. In turn, it follows that $(N^*, C^*)\in \Phi(R_{-N'}, W', D)$. Further, it holds that $\Phi(R_{-N'},W',D)\subseteq \Phi(R, W',D)$ as introducing new voters can only create more solid coalitions. Thus, if SCR agrees on $W'$ and $D$ for both $R$ and $R_{-N'}$ in the current step, it will still agree on these sets for the next step. Since it initially always holds that $W'=\emptyset$ and $D=C$, it follows that SCR chooses $W$ for both $R$ and $R_{-N'}$.
\end{proof}

\section{Committee Monotonicity and \rjr}
\label{sec:rank-jr}

In light of our positive results so far, it is a natural follow-up question whether the SCR rule---or, more generally, any committee monotone voting rule---also satisfies other forms of proportionality. Unfortunately, we will give a negative answer to this question by showing that no committee voting rule satisfies both committee monotonicity and a proportionality notion called \rjr due to \citet{BrPe23a}. 

In more detail, \citeauthor{BrPe23a} suggested a whole family of proportionality notions inspired by fairness axioms from approval-based committee elections. To explain these axioms, we define the \emph{rank} of a candidate~$c$ in a preference relation $\succsim_i$ as $\rank(\succsim_i,x)=|\{y\in A\colon y\succ_ix\}| + 1$. Thus, the most-preferred candidate in a preference relation has rank $1$, while the least-preferred candidate has rank $m$. Now, \citeauthor{BrPe23a} observe that for each $r\in \{1,\dots, m\}$, we can transform a preference profile $R$ into an approval profile $A(R,r)$ by letting each voter approve the alternatives with a rank of at most~$r$. Based on this insight, proportionality notions for approval-based committee elections can be transferred to ranked preferences by requiring that a committee satisfies the given proportionality axiom for the approval profiles $A(R,r)$ for all $r\in \{1,\dots, m\}$. Applying this approach to justified representation (JR) \citep{Aziz_Brill_Conitzer_Elkind_Freeman_Walsh_2015}, one of the weakest fairness notions in approval-based committee voting, results in the following axiom.  

\begin{definition}[\rjr]
	A committee $W$ satisfies \rjr for a preference profile $R$ and committee size $k$ if for all ranks $r\in \{1,\dots, m\}$ and groups of voters $N'\subseteq N$ such that $|N'|\geq\frac{n}{k}$ and $\bigcap_{i\in N'} \{x\in C : \rank(\succsim_i, x) \leq r\}\neq\emptyset$, it holds that $W\cap \bigcup_{i\in N'} \{x\in C : \rank(\succsim_i, x)\leq r\} \neq\emptyset$.
\end{definition}

It turns out that no committee monotone committee voting rule satisfies \rjr, establishing an impossibility theorem.

\begin{theorem}\label{thm:rjr}
	Even for strict preferences, no committee voting rule satisfies both committee monotonicity and \rjr if $n\geq 4$ and $m \ge \frac{n}{q}+1$, where $q\in\mathbb{N}$ satisfies that $\frac{n}{q}\geq 4$ and $\frac{n}{q}\in\mathbb{N}$.
\end{theorem}
\begin{proof}
	Fix some $n$, let $q\in\mathbb{N}$ denote a divisor of $n$ such that $\frac{n}{q}\geq 4$, and define $\ell=\frac{n}{q}$. Consider a set of at least $\ell+ 1$ candidates, the first $\ell + 1$ of which we label as $c_1, \dots, c_\ell$ and $d$. Let $f$ be a committee voting rule that satisfies \rjr and consider the following profile $R$ with $n$ voters:
	\[
		\begin{array}{ll}
			\text{Voters $1,\dots, q$:} \: & c_1 \spref d \spref \cdots \\
			\text{Voters $q+1,\dots, 2q$:} \: & c_2 \spref d \spref \cdots \\
			\text{Voters $2q+1,\dots, 3q$:} \: & c_3 \spref d \spref \cdots \\
			\vdots & \vdots \\
			\text{Voters $n-q+1,\dots, n$:} \: & c_\ell \spref d \spref \cdots
		\end{array}
	\]
	Here, the rankings of the voters can be completed in an arbitrary way.
	
	We will now analyze the committees selected by $f$ for the committee sizes $k = 2$ and $k = \ell$.
	
	We first claim that $d\in f(R,2)$. Assume for a contradiction that $d\not\in f(R,2)$. Since $k = 2$, at most two of the groups in our profile receive their top-ranked candidate in $f(R,2)$. Choose $\lceil \frac{n}{2} \rceil$ voters who do not receive their top-ranked candidate (such voters exist as $\ell \ge 4$). Then \rjr for $r=2$ requires that one of the top-ranked candidates of these voters is chosen, as all of them rank $d$ second but $d\not\in f(R,2)$, a contradiction. Hence, it holds that $d\in f(R,2)$. 
	
	Next, it is easy to see that $f(R,\ell)=\{c_1,c_2,c_3,\dots,c_\ell\}$.
	This follows because, when $k = \ell$, \rjr for $r = 1$ requires that the top-ranked candidate of each group of voters is selected.
	
	This implies that $f$ fails committee monotonicity because $d\in f(R,2)$ but $d\not\in f(R,\ell)$.
\end{proof}

Besides \rjr, \citet{BrPe23a} also adopt several other approval-based fairness axioms to ranked preferences. However, \rjr is their weakest fairness notion for ranked preferences. Thus, showing that no committee monotone voting rule satisfies \rjr implies that, for ranked preferences, no committee monotone rule satisfies any of the proportionality conditions of \citeauthor{BrPe23a}. This includes the property of \emph{Rank-PJR} which strengthens (Hare-)PSC (see \Cref{fig:results} for an overview of implication relationships). Thus, our existence theorems (\Cref{thm:revseq} and \Cref{thm:SCR-ipsc-and-monotone}) cannot be strengthened from (Hare-)PSC to Rank-PJR. 

\section{Application to Rank Aggregation}

As our last contribution, we will discuss the consequences of our results for the problem of rank aggregation. Like in the committee voting setting, the rank aggregation problems takes as input a preference profile (or, more aptly, a ranking profile) $R$. However, the goal is now to compute a winning ranking instead of a winning committee. More formally, rank aggregation is done using \emph{social welfare functions (SWFs)}, which are functions that map every strict preference profile $R$ to a single strict winning ranking $f(R)={\rhdsim}\in\mathcal{L}$. Following \citet{LPW24a}, we will typically write $\rhdsim$ for the output ranking of an SWF, whereas $\succsim$ is reserved for input rankings.

It is easy to see that there is a $1$-to-$1$ correspondence between SWFs and committee monotone committee voting rules. To formalize this, we recall that the rank of an alternative $c$ in a preference relation $\succsim$ is $\rank({\succsim}, c) = \lvert \{c' \in C : c' \succ c\}\rvert + 1$. Then, an SWF $f$ can be converted into a committee monotone voting rule $g$ by setting $g(R, k) = \{c\in C\colon \rank(f(R),c)\leq k\}$, i.e., $g$ chooses the top-$k$ candidates with respect to $f(R)$. Similarly, a committee monotone committee voting rule $g$ can be turned into an SWF by ranking the element in $g(R,k+1) \setminus g(R, k)$ in the $k+1$-st position.\footnote{An analogous correspondence holds between irresolute committee voting rules satisfying a suitable generalization of committee monotonicity and irresolute SWFs \citep{BaCo08a,Elkind2017}.
}

Furthermore, \psc naturally extends to a rank-based proportionality notion for SWFs by requiring that for each $k\in \{1,\dots, m\}$, the top-$k$ candidates with respect to the chosen ranking satisfies \psc. 
\begin{definition}[Proportionality for solid coalitions for rankings]
	Given a preference profile $R$, a ranking $\rhdsim$ satisfies Droop-PSC (Hare-PSC) if for every $k \in \{1,\dots, m\}$, the set $\{x\in C\colon {\rank(\rhdsim, x)\leq k}\}$ satisfies Droop-PSC (Hare-PSC) for the profile $R$ and the committee size $k$.
\end{definition}

Analogously to committee voting rules, an SWF satisfies one of these properties if and only if its chosen ranking always satisfies this property.
\begin{example}
	\iflatexml\else
	\begin{wrapstuff}[r,type=figure,width=5cm]
		\centering
		\begin{tikzpicture}
			\node at (0,0) [draw=black!30] (step1) {	
				\votermultiplicity{$v_1$}{\weakorder{{a},{b},{c},{d},{e}}}
				\votermultiplicity{$v_2$}{\weakorder{{a},{b},{c},{d},{e}}}
				\votermultiplicity{$v_3$}{\weakorder{{a},{b},{c},{d},{e}}}
				\votermultiplicity{$v_4$}{\weakorder{{c},{b},{a},{d},{e}}}
				\votermultiplicity{$v_5$}{\weakorder{{e},{d},{a},{b}, {c}}}
				\votermultiplicity{$v_6$}{\weakorder{{e},{d},{a},{b}, {c}}}
				\votermultiplicity{$v_7$}{\weakorder{{e},{d},{a},{b}, {c}}}
				\votermultiplicity{$v_8$}{\weakorder{{d},{b},{e},{c}, {a}}}
			};
		\end{tikzpicture}
		\vspace{-7pt}
		\captionof*{figure}{Fig. \ref{fig:psc_exp}. An example for PSC.}
	\end{wrapstuff}\fi
	Consider again the profile shown in \Cref{fig:psc_exp}. On this instance, Hare-PSC requires that one of the top-2 candidates in the ranking is $a$, $b$, or $c$, and that $a$ and $e$ are among the top-3. Further, for $k = 4$, the quota is $2$ and therefore at least $2$ of $a$, $b$, or $c$ need to be among the top-4. Hence, a possible Hare-PSC ranking would be 
	$a \succ b \succ e \succ c \succ d$.
	This ranking, however, fails Droop-PSC because, for $k = 2$, the committee $\{a,b\}$ does not satisfy Droop-PSC.
\end{example}

It is easy to see that all our positive results obtained in the previous sections carry over to SWFs, e.g., the SCR rule also satisfies \psc when viewed as an SWF. 

Thus far, we have used PSC as our criterion for proportionality. In the context of rank aggregation, however, there are alternative approaches for measuring the proportionality of a ranking. 
In particular, \citet{LPW24a} propose to analyze the worst-case swap distance of an input ranking to the output ranking, as a function of the fraction of voters that  report the given ranking. 
To make this more formal, we define the swap distance between two ranking $\succsim$ and $\rhdsim$ by $\swap(\succsim, \rhdsim)=|\{(x,y)\in C^2\colon x\succ y\text{ and } y\rhd x\}|$, i.e., the number of pairs the rankings disagree on. 
The central idea of \citeauthor{LPW24a} is that, if an $\alpha$ fraction of the voters report a ranking $\succsim$, then its swap distance to the output ranking should ideally be $(1-\alpha){m\choose 2}$ since $m\choose 2$ is the maximal swap distance between two rankings. 
In other words, in the ideal case, every $\alpha$ fraction of the voters that report the same ranking would agree with the output ranking on at least $\alpha{m\choose 2}$ pairs. 
With this goal in mind, \citeauthor{LPW24a} suggest the Squared Kemeny rule and show that, while this rule does not satisfy this ideal representation guarantee, it satisfies an approximate variant of it. 
In more detail, these authors prove that, if a ranking $\succsim$ is reported by at least $\alpha\cdot n$ voters, the maximum swap distance between $\succsim$ and a ranking $\rhdsim$ chosen by the Squared Kemeny rule is at most $\swap(\succsim,\rhdsim)\leq \sqrt{(1-\alpha)/\alpha} \cdot {m\choose 2}$. 

We next aim to give similar bounds for rankings that satisfy \psc. Specifically, we will show that, up to a small constant factor, the swap distance of a ranking that is reported by an $\alpha$ fraction of the voters and a ranking satisfying \psc is linear in $\alpha$. For instance, this means that the SCR rule gives strong fairness guarantees simultaneously in terms of the swap distance and in terms of PSC. Interestingly, our argument will show that Droop-\psc results in a significantly better bound than Hare-\psc, which can be seen as a first formal comparison of these two concepts. 

\begin{restatable}{theorem}{swapdistance}
	\label{thm:swap}
	Let $R$ be a ranking profile in which at least $\alpha\cdot n$ voters report $\succsim$. Further, let $\rhdsim_H$ be a ranking satisfying Hare-PSC for $R$ and $\rhdsim_D$ a ranking satisfying Droop-PSC for $R$. It holds that
	\begin{enumerate}[leftmargin=*, itemsep=4pt, topsep=4pt, label=(\roman*)]
		\item $swap(\succsim, \rhdsim_H) \le (1 - \alpha + \frac{1+\alpha}{m}) {m \choose 2}$ and 
		\item $swap(\succsim, \rhdsim_D) \le (1 - \alpha)(\frac{m}{m-1}){m \choose 2}$.
	\end{enumerate}
\end{restatable}
\begin{proof}
	We will focus here on the proof for Hare-PSC; the proof for Droop-PSC follows the same strategy and can be found in \Cref{app:swap}. Now, let ${\succsim} = c_1 \succ c_2 \succ \dots\succ c_m$ and fix a profile~$R$. Further, let $N'$ denote the group of voters reporting $\succsim$ and suppose that $|N'|\geq \alpha\cdot n$ for some $\alpha\in [0,1]$. For every rank $r \in [m]$, the set $N'$ forms a solid coalition supporting the set $\{c_1, \dots, c_r\}$. Because $\lvert N'\rvert \ge \alpha \cdot n$, Hare-PSC requires that all $r$ of these candidates are selected when the committee size $k$ satisfies that $\alpha\geq \frac{r}{k}$. 
	Next, let $\rhdsim_H$ denote a ranking that satisfies Hare-PSC for~$R$. From the previous observation, we obtain that all of $c_1, \dots, c_r$ are all ranked among the first $\lceil \frac{r}{\alpha}\rceil$ candidates in $\rhdsim_H$. Now, let $x$ be a candidate ranked by $\rhdsim_H$ at rank $t$ with $\lceil \frac{r}{\alpha}\rceil + 1 \le t\le \lceil \frac{r+1}{\alpha}\rceil$ for some $r\in \{0,\dots, m\}$ such that $\lceil \frac{r}{\alpha}\rceil\leq m-1$. It holds that $x=c_i$ for some candidate $i$ with $i>r$ as otherwise Hare-PSC is violated. Consequently, candidate $x$ can have at most $t-r-1$ inversions to candidates ranked before it in $\rhdsim_H$, as none of the candidates between $c_1$ to $c_r$ incur an inversion. Let $r'$ be the last rank for which $\lceil r'/\alpha\rceil \le m-1$, i.e., $r' = \lfloor \alpha(m-1)\rfloor$. By our previous observation, we can upper bound the number of inversions of $\rhdsim_H$ by
	\begin{align*}
		\swap(\succsim, \rhdsim_H) & \le \sum_{i = 1}^{\lceil 1/\alpha\rceil}(i-1) + \sum_{i = \lceil 1/\alpha\rceil+1 }^{\lceil 2/\alpha\rceil}(i - 2) + \dots +  \sum_{i = \lceil r'/\alpha\rceil+1}^{m}(i - 1-r') \\ 
		&= \sum_{i = 0}^{\lceil 1/\alpha\rceil - 1}(i) + \sum_{i = \lceil 1/\alpha\rceil }^{\lceil 2/\alpha\rceil-1}(i - 1) + \dots + \sum_{i = \lceil r'/\alpha\rceil}^{m-1}(i - r') 
	\end{align*}
	
	In this sum, the term $\lceil i/\alpha\rceil - i$ for $i\in \{1,\dots, r'\}$ appears twice, namely as the last term of the $i$-th sum and as the first term of the $(i+1)$-th sum. By contrast, each other term appears exactly once. Therefore, we get that 
	\begin{align*}
		\swap(\succsim, \rhdsim_H) &\le \sum_{i = 0}^{m - r' - 1} (i) + \sum_{i = 1}^{r'} (\lceil i/\alpha\rceil - i)	\\
		&\leq \sum_{i = 0}^{m - r' - 1} (i)  + (1/\alpha-1)\sum_{i = 1}^{r'} (i) + r' \\
		&= 	(m - r')(m - r'-1)/2 + (1/\alpha - 1)(r'+1)(r')/2 + r'\\
		&=  (r')^2/(2\alpha) + (1 + \tfrac{1}{2\alpha} - m)r' + m^2/2 - m/2.
	\end{align*}
	
	Next, let $f(x) = x^2/(2\alpha) + (1 + \frac{1}{2\alpha} - m)x + m^2/2 - m/2$. The derivative of this function is $\frac{x}{\alpha} + (1 + \frac{1}{2\alpha} - m)$, so we deduce that $f$ has a minimum point at $x' = \alpha (m-1) - \frac{1}{2}$. Since $f$ is a quadratic function and thus symmetric around $x'$, we know that $f(x' - \varepsilon') < f(x' + \varepsilon)$ for all $\varepsilon > \varepsilon' > 0$. Now, if $r' = \lfloor \alpha (m-1) \rfloor\ge (\alpha (m-1)) - \frac{1}{2}$, we directly obtain that $f(r') \le f(\alpha(m-1))$ since $r'\leq \alpha(m-1)$. On the other hand, if $r' < (\alpha (m-1)) - \frac{1}{2}$, we define $\delta = (\alpha (m-1)) - \frac{1}{2} - r' < \frac{1}{2}$. By our previous observation, we also get that $f(r') = f(x' - \delta) < f(x' + \frac{1}{2}) = f(\alpha (m-1))$. Hence, in both cases, it holds that $f(r') \le f(\alpha (m-1))$, so we get that
	\begin{align*}
		\swap(\succsim, \rhdsim_H) &\leq f(r') \le f(\alpha(m-1))\\
		&=  \alpha^2(m-1)^2/(2\alpha) + (1+\tfrac{1}{2\alpha}-m)\alpha(m-1)+m^2/2-m/2\\
		&= \left(1 - \alpha + \frac{1 + \alpha}{m}\right){m \choose 2}.
	\end{align*}
	This proves Claim (i) of our theorem.
\end{proof}

\begin{figure}[t]
	\centering
	\iflatexml
	\includegraphics[width=\linewidth]{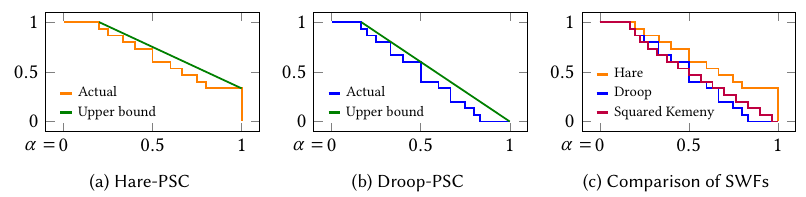}
	\else
	\begin{subfigure}{0.32\textwidth}
	\begin{tikzpicture}
		\begin{axis}[
			width=5.2cm,
			height=3.6cm,
			legend style={shortl},
			legend style={
				cells={anchor=west},
				legend pos=south west,
				draw=none,
				fill=none,
				row sep=-1pt,
				font=\scriptsize
			},
			]
			\addplot [orange, line width=0.03cm,
			] coordinates {
				(0.0,1.0) (0.2, 1) (0.2, 14/15) (0.25, 14/15) (0.25, 13/15) (0.333, 13/15) (0.333, 12/15) (0.4, 12/15)
				(0.4, 11/15) (0.5, 11/15) (0.5, 10/15) (0.5, 10/15) (0.5, 9/15) (0.6, 9/15) (0.6, 8/15) (0.667, 8/15) (0.667, 7/15) (0.75,  7/15) (0.75 ,6/15) (0.8, 6/15) (0.8, 5/15) (1, 5/15) (1,0)
			};
			
			\addplot[color=green!50!black, line width=0.03cm, domain = 0.2:1]{(1-x) + (1+x)/6};
			\legend{Actual, Upper bound}
		\end{axis}
		\node at (-0.1,-0.25) {$\alpha = $};
	\end{tikzpicture}
	\caption{Hare-PSC}
	\end{subfigure}
	\begin{subfigure}{0.32\textwidth}
	\begin{tikzpicture}
		\begin{axis}[
			width=5.2cm,
			height=3.6cm,
			legend style={shortl},
			legend style={
				cells={anchor=west},
				legend pos=south west,
				draw=none,
				fill=none,
				row sep=-1pt,
				font=\scriptsize
			},
			]
			\addplot [blue, line width=0.03cm,
			] coordinates {
				(0.0,1.0) (0.167, 1) (0.167, 14/15) (0.2, 14/15) (0.2, 13/15) (0.25, 13/15) (0.25, 12/15) (0.33, 12/15)
				(0.33, 11/15) (0.33, 10/15) (0.4, 10/15) (0.4, 9/15) (0.5, 9/15) (0.5, 8/15) (0.5, 7/15)  (0.5, 6/15) (0.6, 6/15) (0.6, 5/15) (0.667, 5/15) (0.667, 3/15) (0.75, 3/15) (0.75, 2/15) (0.8, 2/15) (0.8, 1/15) (0.833, 1/15) (0.833, 0) (1,0)
			};
			
			\addplot[color=green!50!black,line width=0.03cm, domain = 0.167:1]{min(1, (1 - x)*6/5)};
			\legend{Actual, Upper bound}
		\end{axis}
		\node at (-0.1,-0.25) {$\alpha = $};
	\end{tikzpicture}
	\caption{Droop-PSC}
	\end{subfigure}
	\begin{subfigure}{0.32\textwidth}
	\begin{tikzpicture}
		\begin{axis}[
			width=5.2cm,
			height=3.6cm,
			legend style={shortl},
			legend style={
				cells={anchor=west},
				legend pos=south west,
				draw=none,
				fill=none,
				row sep=-1pt,
				font=\scriptsize
			},
			]
			\addplot [orange, line width=0.03cm,
			] coordinates {
				(0.0,1.0) (0.2, 1) (0.2, 14/15) (0.25, 14/15) (0.25, 13/15) (0.333, 13/15) (0.333, 12/15) (0.4, 12/15)
				(0.4, 11/15) (0.5, 11/15) (0.5, 10/15) (0.5, 10/15) (0.5, 9/15) (0.6, 9/15) (0.6, 8/15) (0.667, 8/15) (0.667, 7/15) (0.75,  7/15) (0.75 ,6/15) (0.8, 6/15) (0.8, 5/15) (1, 5/15) (1,0)
			};
			\addplot [blue, line width=0.03cm,
			] coordinates {
				(0.0,1.0) (0.167, 1) (0.167, 14/15) (0.2, 14/15) (0.2, 13/15) (0.25, 13/15) (0.25, 12/15) (0.33, 12/15)
				(0.33, 11/15) (0.33, 10/15) (0.4, 10/15) (0.4, 9/15) (0.5, 9/15) (0.5, 8/15) (0.5, 7/15)  (0.5, 6/15) (0.6, 6/15) (0.6, 5/15) (0.667, 5/15) (0.667, 3/15) (0.75, 3/15) (0.75, 2/15) (0.8, 2/15) (0.8, 1/15) (0.833, 1/15) (0.833, 0) (1,0)
			};
			\addplot [purple, line width=0.03cm,
			] coordinates {
				(0.0,1.0) (0.171, 1) (0.171, 14/15) (0.196, 14/15) 
				(0.196, 13/15) (0.226, 13/15)
				(0.226, 12/15) (0.271, 12/15)
				(0.271, 11/15) (0.323, 11/15)
				(0.323, 10/15) (0.377, 10/15)
				(0.377, 9/15) (0.438, 9/15)
				(0.438, 8/15) (0.5, 8/15)
				(0.5, 7/15) (0.567, 7/15)
				(0.567, 6/15) (0.633, 6/15)
				(0.633, 5/15) (0.7, 5/15)
				(0.7, 4/15) (0.767, 4/15)
				(0.767, 3/15) (0.833, 3/15)
				(0.833, 2/15) (0.9, 2/15)
				(0.9, 1/15) (0.967, 1/15) (0.967, 0)(1,0)
			};
			
			\legend{Hare, Droop, Squared Kemeny}
		\end{axis}
		\node at (-0.1,-0.25) {$\alpha = $};
	\end{tikzpicture}
	\caption{Comparison of SWFs}
	\label{fig:intro-alpha-curve-c}
	\end{subfigure}
	\fi
	\vspace{-7pt}
	\caption{The maximum (normalized) swap distance between the output rankings for various rules to an input ranking, as a function of the fraction of the voters that report the input ranking. The swap distance is normalized to be in $[0,1]$. The left figure shows the maximum swap distance between an input ranking of weight $\alpha$ and a ranking that satisfies Hare-PSC. Specifically, we display both our theoretical upper bound (green) and the actual upper bound computed for $m=6$ (orange). Similarly, in the plot in the middle, we show an analogous plot for rankings satisfying Droop-PSC, with green corresponding to our theoretical upper bound and blue being the actual upper bound for $m=6$. Finally, on the right figure, we compare the maximum swap distance to an input ranking of weight $\alpha$ for the Squared Kemeny rule (purple), rankings satisfying Droop-PSC (blue), and rankings satisfying Hare-PSC (orange).}	\label{fig:intro-alpha-curve}
\end{figure}

To gain a better intuition for these bounds, we refer the reader to \Cref{fig:intro-alpha-curve}, where we plotted our theoretical upper bounds as well as the actual worst-case bound for $m = 6$ alternatives with respect to Hare-PSC and Droop-PSC. As we can see, Droop-PSC gives significantly better bounds. In particular, when $\alpha$ is close (but not equal) to $1$, Hare-PSC is not sufficient to determine the first candidate in the ranking and even the bottom-ranked candidate with respect to our input ranking can be top-ranked in the output ranking, whereas Droop-PSC precludes this behavior. Furthermore, in \iflatexml \Cref{fig:intro-alpha-curve}(c)\else \Cref{fig:intro-alpha-curve-c}\fi, we compare the worst-case rankings that satisfy Hare-PSC and Droop-PSC to the worst-case ranking that can be chosen by the Squared Kemeny rule. 
As we can see, rankings that satisfy Droop-PSC perform similarly well as the ones chosen by the Squared Kemeny rule. Since the latter rule is NP-hard to compute \citep{LPW24a} while rules like SCR run in polynomial time, this provides further motivation for using methods that satisfy our fairness conditions.

\section{Conclusions}

In this paper, we present the first committee voting rules for ranked preferences that satisfy committee monotonicity and proportionality for solid coalitions (\psc). Specifically, we first show that such rules can be defined based on known rules that satisfy \psc by running these rules in a reverse sequential mode. However, this approach results in rather unnatural and hard to analyze rules, so we further design a new committee voting rule called the Solid Coalition Refinement rule that simultaneously satisfies committee monotonicity and \psc. This rule is inspired by the D'Hondt method for apportionment, and it can even be generalized to weak preferences, for which it satisfies a strengthening of \psc called inclusion \psc. Using this rule, we also disprove a conjecture by \citet{GJM24a} regarding the compatibility of \psc and a notion called independence of losing voter blocs for truncated preferences. Additionally, we show that committee monotonicity is not compatible with a recently suggested family of fairness concepts by \citet{BrPe23a}. Finally, we re-examine our results in the context of rank aggregation and provide bounds on the maximum swap distance between a given input ranking and a PSC ranking depending on the fraction of voters that report the input ranking.
 
A natural follow-up question is whether our positive results extend to the more general setting of participatory budgeting (PB). 
In this setting, each candidate has a cost and the task is to choose a set of candidates within a prescribed budget. 
For PB, committee monotonicity has been generalized to an axiom called \emph{limit monotonicity} \citep{Talmon2019} and \ipsc has been defined for this setting by \citet{azizlee2021}. However, a simple counterexample shows that limit monotonicity is incompatible with \ipsc.
First, we note that \ipsc implies \emph{exhaustiveness} (the left-over budget cannot be enough to purchase another candidate).
Now, assume that there are two voters and three candidates $c_1, c_2, c_3$ where $c_1$ and $c_2$ have cost $2$ and $c_3$ has cost 1. If $c_1$ and $c_2$ are each ranked first by one of the voters, these candidates must be chosen when the budget is 4.
However, if the budget is 3, exhaustiveness requires $c_3$ to be chosen.
This violates limit monotonicity, so weaker variants of proportionality or committee monotonicity are required to obtain positive results in PB.

Our results also give other directions for future work. In particular, we leave it open whether there is a polynomial-time computable rule for weak preferences that satisfies the axiomatic properties of SCR.
Moreover, much remains unknown about the compatibility of committee monotonicity and proportionality under approval preferences \citep[see also][]{LacknerSkowrongABCvotingBook2023}.

\section*{Acknowledgements}

We are grateful to Tomasz W\k{a}s for advice and for pointing out connections to rank aggregation.
This work was supported by the NSF-CSIRO grant on “Fair Sequential Collective Decision-Making" (RG230833).
Dominik Peters was funded in part by the Agence Nationale de la Recherche
under grant ANR22-CE26-0019 (CITIZENS) and as part of the France 2030 program under
grant ANR-23-IACL-0008 (PR[AI]RIE-PSAI).

\bibliography{./pubs,./paper}


\begin{thebibliography}{40}


\ifx \showCODEN    \undefined \def \showCODEN     #1{\unskip}     \fi
\ifx \showDOI      \undefined \def \showDOI       #1{#1}\fi
\ifx \showISBNx    \undefined \def \showISBNx     #1{\unskip}     \fi
\ifx \showISBNxiii \undefined \def \showISBNxiii  #1{\unskip}     \fi
\ifx \showISSN     \undefined \def \showISSN      #1{\unskip}     \fi
\ifx \showLCCN     \undefined \def \showLCCN      #1{\unskip}     \fi
\ifx \shownote     \undefined \def \shownote      #1{#1}          \fi
\ifx \showarticletitle \undefined \def \showarticletitle #1{#1}   \fi
\ifx \showURL      \undefined \def \showURL       {\relax}        \fi
\providecommand\bibfield[2]{#2}
\providecommand\bibinfo[2]{#2}
\providecommand\natexlab[1]{#1}
\providecommand\showeprint[2][]{arXiv:#2}

\bibitem[Arrow(1950)]%
        {Arrow1950}
\bibfield{author}{\bibinfo{person}{Kenneth~J. Arrow}.}
  \bibinfo{year}{1950}\natexlab{}.
\newblock \showarticletitle{A Difficulty in the Concept of Social Welfare}.
\newblock \bibinfo{journal}{\emph{Journal of Political Economy}}
  \bibinfo{volume}{58}, \bibinfo{number}{4} (\bibinfo{year}{1950}),
  \bibinfo{pages}{328--346}.
\newblock
\urldef\tempurl%
\url{https://doi.org/10.1086/256963}
\showDOI{\tempurl}


\bibitem[Aziz(2017)]%
        {Aziz17a}
\bibfield{author}{\bibinfo{person}{Haris Aziz}.}
  \bibinfo{year}{2017}\natexlab{}.
\newblock \bibinfo{title}{A Note on Justified Representation Under the Reverse
  Sequential {PAV} rule}.  (\bibinfo{year}{2017}).
\newblock
\newblock
\shownote{Unpublished manuscript. Available at
  \url{http://www.cse.unsw.edu.au/~haziz/invseqpav.pdf}}.


\bibitem[Aziz et~al\mbox{.}(2017)]%
        {Aziz_Brill_Conitzer_Elkind_Freeman_Walsh_2015}
\bibfield{author}{\bibinfo{person}{Haris Aziz}, \bibinfo{person}{Markus Brill},
  \bibinfo{person}{Vincent Conitzer}, \bibinfo{person}{Edith Elkind},
  \bibinfo{person}{Rupert Freeman}, {and} \bibinfo{person}{Toby Walsh}.}
  \bibinfo{year}{2017}\natexlab{}.
\newblock \showarticletitle{Justified Representation in Approval-Based
  Committee Voting}.
\newblock \bibinfo{journal}{\emph{Social Choice and Welfare}}
  \bibinfo{volume}{48}, \bibinfo{number}{2} (\bibinfo{year}{2017}),
  \bibinfo{pages}{461--485}.
\newblock


\bibitem[Aziz and Lee(2020)]%
        {Aziz2019}
\bibfield{author}{\bibinfo{person}{Haris Aziz} {and} \bibinfo{person}{Barton~E.
  Lee}.} \bibinfo{year}{2020}\natexlab{}.
\newblock \showarticletitle{The Expanding Approvals Rule: Improving
  proportional representation and monotonicity}.
\newblock \bibinfo{journal}{\emph{Social Choice and Welfare}}
  \bibinfo{volume}{54}, \bibinfo{number}{1} (\bibinfo{year}{2020}),
  \bibinfo{pages}{1--45}.
\newblock
\urldef\tempurl%
\url{https://doi.org/10.1007/s00355-019-01208-3}
\showDOI{\tempurl}


\bibitem[Aziz and Lee(2021)]%
        {azizlee2021}
\bibfield{author}{\bibinfo{person}{Haris Aziz} {and} \bibinfo{person}{Barton~E.
  Lee}.} \bibinfo{year}{2021}\natexlab{}.
\newblock \showarticletitle{Proportionally Representative Participatory
  Budgeting with Ordinal Preferences}.
\newblock \bibinfo{journal}{\emph{Proceedings of the 35th AAAI Conference on
  Artificial Intelligence (AAAI)}} (\bibinfo{year}{2021}),
  \bibinfo{pages}{5110--5118}.
\newblock
\urldef\tempurl%
\url{https://doi.org/10.1609/aaai.v35i6.16646}
\showDOI{\tempurl}


\bibitem[Aziz and Lee(2022)]%
        {AZIZ2022248}
\bibfield{author}{\bibinfo{person}{Haris Aziz} {and} \bibinfo{person}{Barton~E.
  Lee}.} \bibinfo{year}{2022}\natexlab{}.
\newblock \showarticletitle{A characterization of proportionally representative
  committees}.
\newblock \bibinfo{journal}{\emph{Games and Economic Behavior}}
  \bibinfo{volume}{133} (\bibinfo{year}{2022}), \bibinfo{pages}{248--255}.
\newblock
\urldef\tempurl%
\url{https://doi.org/10.1016/j.geb.2022.03.006}
\showDOI{\tempurl}


\bibitem[Balinski and Young(2001)]%
        {FairRepresentationBook}
\bibfield{author}{\bibinfo{person}{Michel Balinski} {and}
  \bibinfo{person}{H.~Peyton Young}.} \bibinfo{year}{2001}\natexlab{}.
\newblock \bibinfo{booktitle}{\emph{Fair Representation: Meeting the Ideal of
  One Man, One Vote} (\bibinfo{edition}{2nd} ed.)}.
\newblock \bibinfo{publisher}{Brookings Institution Press}.
\newblock


\bibitem[Barber{\`a} and Coelho(2008)]%
        {BaCo08a}
\bibfield{author}{\bibinfo{person}{Salvador Barber{\`a}} {and}
  \bibinfo{person}{Danilo Coelho}.} \bibinfo{year}{2008}\natexlab{}.
\newblock \showarticletitle{How to choose a non-controversial list with k
  names}.
\newblock \bibinfo{journal}{\emph{Social Choice and Welfare}}
  \bibinfo{volume}{31}, \bibinfo{number}{1} (\bibinfo{year}{2008}),
  \bibinfo{pages}{79--96}.
\newblock
\urldef\tempurl%
\url{https://doi.org/10.1007/s00355-007-0268-6}
\showDOI{\tempurl}


\bibitem[{Bartholdi, III} and Orlin(1991)]%
        {BartholdiSTV1991}
\bibfield{author}{\bibinfo{person}{John~J. {Bartholdi, III}} {and}
  \bibinfo{person}{James~B. Orlin}.} \bibinfo{year}{1991}\natexlab{}.
\newblock \showarticletitle{Single transferable vote resists strategic voting}.
\newblock \bibinfo{journal}{\emph{Social Choice and Welfare}}
  \bibinfo{volume}{8}, \bibinfo{number}{4} (\bibinfo{year}{1991}),
  \bibinfo{pages}{341--354}.
\newblock
\urldef\tempurl%
\url{https://doi.org/10.1007/BF00183045}
\showDOI{\tempurl}


\bibitem[Black(1958)]%
        {Black1958}
\bibfield{author}{\bibinfo{person}{Duncan Black}.}
  \bibinfo{year}{1958}\natexlab{}.
\newblock \bibinfo{booktitle}{\emph{The Theory of Committees and Elections}}.
\newblock \bibinfo{publisher}{Cambridge University Press}.
\newblock
\urldef\tempurl%
\url{https://doi.org/10.1007/978-94-009-4225-7}
\showDOI{\tempurl}


\bibitem[Brill et~al\mbox{.}(2024)]%
        {BGP+24a}
\bibfield{author}{\bibinfo{person}{Markus Brill}, \bibinfo{person}{Paul
  G{\"o}lz}, \bibinfo{person}{Dominik Peters}, \bibinfo{person}{Ulrike
  Schmidt-Kraepelin}, {and} \bibinfo{person}{Kai Wilker}.}
  \bibinfo{year}{2024}\natexlab{}.
\newblock \showarticletitle{Approval-Based Apportionment}.
\newblock \bibinfo{journal}{\emph{Mathematical Programming}}
  \bibinfo{volume}{203}, \bibinfo{number}{1--2} (\bibinfo{year}{2024}),
  \bibinfo{pages}{77--105}.
\newblock
\urldef\tempurl%
\url{https://doi.org/10.1007/s10107-022-01852-1}
\showDOI{\tempurl}


\bibitem[Brill and Peters(2023)]%
        {BrPe23a}
\bibfield{author}{\bibinfo{person}{Markus Brill} {and} \bibinfo{person}{Jannik
  Peters}.} \bibinfo{year}{2023}\natexlab{}.
\newblock \showarticletitle{Robust and Verifiable Proportionality Axioms for
  Multiwinner Voting}. In \bibinfo{booktitle}{\emph{Proceedings of the 24th ACM
  Conference on Economics and Computation (ACM-EC)}}. \bibinfo{publisher}{ACM
  Press}, \bibinfo{pages}{301}.
\newblock
\urldef\tempurl%
\url{https://doi.org/10.1145/3580507.3597785}
\showDOI{\tempurl}
\newblock
\shownote{Full version:
  arXiv:\href{https://arxiv.org/abs/2302.01989}{2302.01989}}.


\bibitem[Chamberlin and Courant(1983)]%
        {ChCo83a}
\bibfield{author}{\bibinfo{person}{John~R. Chamberlin} {and}
  \bibinfo{person}{Paul~N. Courant}.} \bibinfo{year}{1983}\natexlab{}.
\newblock \showarticletitle{Representative Deliberations and Representative
  Decisions: {P}roportional Representation and the {B}orda Rule}.
\newblock \bibinfo{journal}{\emph{The American Political Science Review}}
  \bibinfo{volume}{77}, \bibinfo{number}{3} (\bibinfo{year}{1983}),
  \bibinfo{pages}{718--733}.
\newblock
\urldef\tempurl%
\url{https://doi.org/10.2307/1957270}
\showDOI{\tempurl}


\bibitem[Delemazure and Peters(2024)]%
        {DePe24a}
\bibfield{author}{\bibinfo{person}{Theo Delemazure} {and}
  \bibinfo{person}{Dominik Peters}.} \bibinfo{year}{2024}\natexlab{}.
\newblock \showarticletitle{Generalizing Instant Runoff Voting to Allow
  Indifferences}. In \bibinfo{booktitle}{\emph{Proceedings of the 25th ACM
  Conference on Economics and Computation (ACM-EC)}}. \bibinfo{publisher}{ACM
  Press}, \bibinfo{pages}{50}.
\newblock
\urldef\tempurl%
\url{https://doi.org/10.1145/3670865.3673501}
\showDOI{\tempurl}
\newblock
\shownote{Full version:
  arXiv:\href{https://arxiv.org/abs/2404.11407}{2404.11407}}.


\bibitem[Dummett(1984)]%
        {Dummett1984}
\bibfield{author}{\bibinfo{person}{Michael Dummett}.}
  \bibinfo{year}{1984}\natexlab{}.
\newblock \bibinfo{booktitle}{\emph{Voting Procedures}}.
\newblock \bibinfo{publisher}{Oxford University Press}.
\newblock


\bibitem[Elkind et~al\mbox{.}(2017)]%
        {Elkind2017}
\bibfield{author}{\bibinfo{person}{Edith Elkind}, \bibinfo{person}{Piotr
  Faliszewski}, \bibinfo{person}{Piotr Skowron}, {and} \bibinfo{person}{Arkadii
  Slinko}.} \bibinfo{year}{2017}\natexlab{}.
\newblock \showarticletitle{Properties of multiwinner voting rules}.
\newblock \bibinfo{journal}{\emph{Social Choice and Welfare}}
  \bibinfo{volume}{48}, \bibinfo{number}{3} (\bibinfo{year}{2017}),
  \bibinfo{pages}{599--632}.
\newblock
\urldef\tempurl%
\url{https://doi.org/10.1007/s00355-017-1026-z}
\showDOI{\tempurl}


\bibitem[Faliszewski et~al\mbox{.}(2017)]%
        {FaliszewskiTrends2017}
\bibfield{author}{\bibinfo{person}{Piotr Faliszewski}, \bibinfo{person}{Piotr
  Skowron}, \bibinfo{person}{Arkadii Slinko}, {and} \bibinfo{person}{Nimrod
  Talmon}.} \bibinfo{year}{2017}\natexlab{}.
\newblock \showarticletitle{Multiwinner Voting: A New Challenge for Social
  Choice Theory}.
\newblock In \bibinfo{booktitle}{\emph{Trends in Computational Social Choice}},
  \bibfield{editor}{\bibinfo{person}{Ulle Endriss}} (Ed.). Chapter~2.
\newblock
\urldef\tempurl%
\url{https://archive.illc.uva.nl/COST-IC1205/BookDocs/Chapters/TrendsCOMSOC-02.pdf}
\showURL{%
\tempurl}


\bibitem[Faliszewski et~al\mbox{.}(2019)]%
        {FSST19a}
\bibfield{author}{\bibinfo{person}{Piotr Faliszewski}, \bibinfo{person}{Piotr
  Skowron}, \bibinfo{person}{Arkadii Slinko}, {and} \bibinfo{person}{Nimrod
  Talmon}.} \bibinfo{year}{2019}\natexlab{}.
\newblock \showarticletitle{Committee Scoring Rules: Axiomatic Characterization
  and Hierachy}.
\newblock \bibinfo{journal}{\emph{ACM Transactions on Economics and
  Computation}} \bibinfo{volume}{7}, \bibinfo{number}{1}
  (\bibinfo{year}{2019}), \bibinfo{pages}{Article 3}.
\newblock
\urldef\tempurl%
\url{https://doi.org/10.1145/3296672}
\showDOI{\tempurl}


\bibitem[Graham-Squire et~al\mbox{.}(2024)]%
        {GJM24a}
\bibfield{author}{\bibinfo{person}{Adam Graham-Squire},
  \bibinfo{person}{Matthew~I. Jones}, {and} \bibinfo{person}{David McCune}.}
  \bibinfo{year}{2024}\natexlab{}.
\newblock \bibinfo{title}{New fairness criteria for truncated ballots in
  multi-winner ranked-choice elections}.
\newblock
\newblock
\showeprint[arxiv]{2408.03926}


\bibitem[Israel and Brill(2025)]%
        {IsBr24a}
\bibfield{author}{\bibinfo{person}{Jonas Israel} {and} \bibinfo{person}{Markus
  Brill}.} \bibinfo{year}{2025}\natexlab{}.
\newblock \showarticletitle{Dynamic proportional rankings}.
\newblock \bibinfo{journal}{\emph{Social Choice and Welfare}}
  \bibinfo{volume}{64} (\bibinfo{year}{2025}), \bibinfo{pages}{221--261}.
\newblock
\urldef\tempurl%
\url{https://doi.org/10.1007/s00355-023-01498-8}
\showDOI{\tempurl}


\bibitem[Janson(2016)]%
        {Jans16a}
\bibfield{author}{\bibinfo{person}{Svante Janson}.}
  \bibinfo{year}{2016}\natexlab{}.
\newblock \bibinfo{title}{Phragm\'en's and Thiele's election methods}.
\newblock
\newblock
\showeprint[arxiv]{1611.08826}


\bibitem[Kamwa(2013)]%
        {Kamwa2013}
\bibfield{author}{\bibinfo{person}{Eric Kamwa}.}
  \bibinfo{year}{2013}\natexlab{}.
\newblock \showarticletitle{The increasing committee size paradox with small
  number of candidates}.
\newblock \bibinfo{journal}{\emph{Economics Bulletin}}  \bibinfo{volume}{33}
  (\bibinfo{year}{2013}), \bibinfo{pages}{967--972}.
\newblock
\urldef\tempurl%
\url{http://www.accessecon.com/Pubs/EB/2013/Volume33/EB-13-V33-I2-P92.pdf}
\showURL{%
\tempurl}


\bibitem[Lackner and Skowron(2023)]%
        {LacknerSkowrongABCvotingBook2023}
\bibfield{author}{\bibinfo{person}{Martin Lackner} {and} \bibinfo{person}{Piotr
  Skowron}.} \bibinfo{year}{2023}\natexlab{}.
\newblock \bibinfo{booktitle}{\emph{Multi-Winner Voting with Approval
  Preferences}}.
\newblock \bibinfo{publisher}{Springer-Verlag}.
\newblock
\urldef\tempurl%
\url{https://doi.org/10.1007/978-3-031-09016-5}
\showDOI{\tempurl}


\bibitem[Lederer et~al\mbox{.}(2024)]%
        {LPW24a}
\bibfield{author}{\bibinfo{person}{Patrick Lederer}, \bibinfo{person}{Dominik
  Peters}, {and} \bibinfo{person}{Tomasz W{\k{a}}s}.}
  \bibinfo{year}{2024}\natexlab{}.
\newblock \showarticletitle{The {Squared} {Kemeny} Rule for Averaging
  Rankings}. In \bibinfo{booktitle}{\emph{Proceedings of the 25th ACM
  Conference on Economics and Computation (ACM-EC)}}. \bibinfo{pages}{755}.
\newblock
\urldef\tempurl%
\url{https://doi.org/10.1145/3670865.3673521}
\showDOI{\tempurl}
\newblock
\shownote{Full version:
  \href{https://arxiv.org/abs/2404.08474}{arXiv:2404.08474}}.


\bibitem[Lu and Boutilier(2011)]%
        {LuBo11d}
\bibfield{author}{\bibinfo{person}{Tyler Lu} {and} \bibinfo{person}{Craig
  Boutilier}.} \bibinfo{year}{2011}\natexlab{}.
\newblock \showarticletitle{Budgeted Social Choice: From Consensus to
  Personalized Decision Making}. In \bibinfo{booktitle}{\emph{Proceedings of
  the 22nd International Joint Conference on Artificial Intelligence (IJCAI)}}.
  \bibinfo{pages}{280--286}.
\newblock
\urldef\tempurl%
\url{https://doi.org/10.5591/978-1-57735-516-8/IJCAI11-057}
\showDOI{\tempurl}


\bibitem[McCune and Graham-Squire(2024)]%
        {McGr23a}
\bibfield{author}{\bibinfo{person}{David McCune} {and} \bibinfo{person}{Adam
  Graham-Squire}.} \bibinfo{year}{2024}\natexlab{}.
\newblock \showarticletitle{Monotonicity Anomalies in {Scottish} Local
  Government Elections}.
\newblock \bibinfo{journal}{\emph{Social Choice and Welfare}}
  \bibinfo{volume}{63}, \bibinfo{number}{1} (\bibinfo{year}{2024}),
  \bibinfo{pages}{69--101}.
\newblock
\urldef\tempurl%
\url{https://doi.org/10.1007/s00355-024-01522-5}
\showDOI{\tempurl}


\bibitem[Monroe(1995)]%
        {Monr95a}
\bibfield{author}{\bibinfo{person}{Burt~L. Monroe}.}
  \bibinfo{year}{1995}\natexlab{}.
\newblock \showarticletitle{Fully Proportional Representation}.
\newblock \bibinfo{journal}{\emph{The American Political Science Review}}
  \bibinfo{volume}{89}, \bibinfo{number}{4} (\bibinfo{year}{1995}),
  \bibinfo{pages}{925--940}.
\newblock
\urldef\tempurl%
\url{https://doi.org/10.2307/2082518}
\showDOI{\tempurl}


\bibitem[Peters(2024)]%
        {Pete24a}
\bibfield{author}{\bibinfo{person}{Dominik Peters}.}
  \bibinfo{year}{2024}\natexlab{}.
\newblock \showarticletitle{Proportional Representation for Artificial
  Intelligence}. In \bibinfo{booktitle}{\emph{Proceedings of the 27th European
  Conference on Artificial Intelligence (ECAI)}}. \bibinfo{pages}{27--31}.
\newblock
\urldef\tempurl%
\url{https://doi.org/10.3233/FAIA240463}
\showDOI{\tempurl}


\bibitem[Peters et~al\mbox{.}(2021)]%
        {PPS21a}
\bibfield{author}{\bibinfo{person}{Dominik Peters}, \bibinfo{person}{Grzegorz
  Pierczy\'{n}ski}, {and} \bibinfo{person}{Piotr Skowron}.}
  \bibinfo{year}{2021}\natexlab{}.
\newblock \showarticletitle{Proportional Participatory Budgeting with Additive
  Utilities}. In \bibinfo{booktitle}{\emph{Advances in Neural Information
  Processing Systems}}, Vol.~\bibinfo{volume}{34}.
  \bibinfo{pages}{12726--12737}.
\newblock
\urldef\tempurl%
\url{https://proceedings.neurips.cc/paper/2021/hash/69f8ea31de0c00502b2ae571fbab1f95-Abstract.html}
\showURL{%
\tempurl}


\bibitem[Ratliff(2003)]%
        {Ratliff2003}
\bibfield{author}{\bibinfo{person}{Thomas~C. Ratliff}.}
  \bibinfo{year}{2003}\natexlab{}.
\newblock \showarticletitle{Some startling inconsistencies when electing
  committees}.
\newblock \bibinfo{journal}{\emph{Social Choice and Welfare}}
  \bibinfo{volume}{21}, \bibinfo{number}{3} (\bibinfo{year}{2003}),
  \bibinfo{pages}{433–454}.
\newblock
\urldef\tempurl%
\url{https://doi.org/10.1007/s00355-003-0209-y}
\showDOI{\tempurl}


\bibitem[Rey and Maly(2023)]%
        {ReyPaBuSurvey2023}
\bibfield{author}{\bibinfo{person}{Simon Rey} {and} \bibinfo{person}{Jan
  Maly}.} \bibinfo{year}{2023}\natexlab{}.
\newblock \bibinfo{booktitle}{\emph{The (Computational) Social Choice Take on
  Indivisible Participatory Budgeting}}.
\newblock \bibinfo{type}{{T}echnical {R}eport}.
\newblock
\showeprint[arxiv]{2303.00621}~[cs.GT]


\bibitem[Schulze(2018)]%
        {Schu18a}
\bibfield{author}{\bibinfo{person}{Markus Schulze}.}
  \bibinfo{year}{2018}\natexlab{}.
\newblock \bibinfo{title}{The {Schulze} Method of Voting}.
\newblock
\newblock
\showeprint[arxiv]{1804.02973}~[cs.GT]


\bibitem[Skowron et~al\mbox{.}(2016)]%
        {SFL16a}
\bibfield{author}{\bibinfo{person}{Piotr Skowron}, \bibinfo{person}{Piotr
  Faliszewski}, {and} \bibinfo{person}{J{\'e}r{\^o}me Lang}.}
  \bibinfo{year}{2016}\natexlab{}.
\newblock \showarticletitle{Finding a collective set of items: From
  proportional multirepresentation to group recommendation}.
\newblock \bibinfo{journal}{\emph{Artificial Intelligence}}
  \bibinfo{volume}{241} (\bibinfo{year}{2016}), \bibinfo{pages}{191--216}.
\newblock
\urldef\tempurl%
\url{https://doi.org/10.1016/j.artint.2016.09.003}
\showDOI{\tempurl}


\bibitem[Skowron et~al\mbox{.}(2017)]%
        {SLB+17a}
\bibfield{author}{\bibinfo{person}{Piotr Skowron}, \bibinfo{person}{Martin
  Lackner}, \bibinfo{person}{Markus Brill}, \bibinfo{person}{Dominik Peters},
  {and} \bibinfo{person}{Edith Elkind}.} \bibinfo{year}{2017}\natexlab{}.
\newblock \showarticletitle{Proportional Rankings}. In
  \bibinfo{booktitle}{\emph{Proceedings of the 26th International Joint
  Conference on Artificial Intelligence (IJCAI)}}. \bibinfo{pages}{409--415}.
\newblock
\urldef\tempurl%
\url{https://doi.org/10.24963/ijcai.2017/58}
\showDOI{\tempurl}


\bibitem[Staring(1986)]%
        {Staring1986}
\bibfield{author}{\bibinfo{person}{Mike Staring}.}
  \bibinfo{year}{1986}\natexlab{}.
\newblock \showarticletitle{Two Paradoxes of Committee Elections}.
\newblock \bibinfo{journal}{\emph{Mathematics Magazine}} \bibinfo{volume}{59},
  \bibinfo{number}{3} (\bibinfo{year}{1986}), \bibinfo{pages}{158–159}.
\newblock
\urldef\tempurl%
\url{https://doi.org/10.1080/0025570X.1986.11977239}
\showDOI{\tempurl}


\bibitem[Talmon and Faliszewski(2019)]%
        {Talmon2019}
\bibfield{author}{\bibinfo{person}{Nimrod Talmon} {and} \bibinfo{person}{Piotr
  Faliszewski}.} \bibinfo{year}{2019}\natexlab{}.
\newblock \showarticletitle{A Framework for Approval-Based Budgeting Methods}.
  In \bibinfo{booktitle}{\emph{Proceedings of the 33rd AAAI Conference on
  Artificial Intelligence (AAAI)}}. \bibinfo{pages}{2181--2188}.
\newblock
\urldef\tempurl%
\url{https://doi.org/10.1609/aaai.v33i01.33012181}
\showDOI{\tempurl}


\bibitem[Tideman(1995)]%
        {Tideman95}
\bibfield{author}{\bibinfo{person}{Nicolaus Tideman}.}
  \bibinfo{year}{1995}\natexlab{}.
\newblock \showarticletitle{The Single Transferable Vote}.
\newblock \bibinfo{journal}{\emph{Journal of Economic Perspectives}}
  \bibinfo{volume}{9}, \bibinfo{number}{1} (\bibinfo{year}{1995}),
  \bibinfo{pages}{27--38}.
\newblock
\urldef\tempurl%
\url{https://doi.org/10.1257/jep.9.1.27}
\showDOI{\tempurl}


\bibitem[Tideman(2017)]%
        {tideman2017collective}
\bibfield{author}{\bibinfo{person}{Nicolaus Tideman}.}
  \bibinfo{year}{2017}\natexlab{}.
\newblock \bibinfo{booktitle}{\emph{Collective Decisions and Voting: The
  Potential for Public Choice}}.
\newblock \bibinfo{publisher}{Routledge}.
\newblock
\urldef\tempurl%
\url{https://doi.org/10.4324/9781315259963}
\showDOI{\tempurl}


\bibitem[Tideman and Richardson(2000)]%
        {TidemanBetterVoting2000}
\bibfield{author}{\bibinfo{person}{Nicolaus Tideman} {and}
  \bibinfo{person}{Daniel Richardson}.} \bibinfo{year}{2000}\natexlab{}.
\newblock \showarticletitle{Better Voting Methods Through Technology: The
  Refinement-Manageability Trade-Off in the Single Transferable Vote}.
\newblock \bibinfo{journal}{\emph{Public Choice}} \bibinfo{volume}{103},
  \bibinfo{number}{1-2} (\bibinfo{year}{2000}), \bibinfo{pages}{13--34}.
\newblock
\urldef\tempurl%
\url{https://doi.org/10.1023/A:1005082925477}
\showDOI{\tempurl}


\bibitem[Woodall(1994)]%
        {Woodall1994}
\bibfield{author}{\bibinfo{person}{Douglas Woodall}.}
  \bibinfo{year}{1994}\natexlab{}.
\newblock \showarticletitle{Properties of preferential election rules}.
\newblock \bibinfo{journal}{\emph{Voting Matters}}  \bibinfo{volume}{3}
  (\bibinfo{year}{1994}).
\newblock
\urldef\tempurl%
\url{https://www.votingmatters.org.uk/issue3/p5.htm}
\showURL{%
\tempurl}


\end{thebibliography}
\clearpage

\appendix

\section{Committee Monotonicity of Quota Preference Score Rules}\label{app:existing-rules-and-comm-mon}

QBS was introduced by \citet{Dummett1984} and later generalized into a class called \emph{Minimal Demand} (MD) rules by \citet{AZIZ2022248}.
For each rank $r$, starting at 1 and increasing one by one, these rules perform two steps:

\begin{enumerate}
	\item Partition the voters into solid coalitions, where two voters are in the same solid coalition if their top $r$ ranked candidates are the same, regardless of the ordering within the top $r$.
	\item Consider each solid coalition $N'$ supporting candidate set $C'$ in the partition. If $N'$ is entitled to more representation under \psc, then additional candidates from $C'$ are elected until this entitlement is met. 
\end{enumerate}

The rules in this family differ by the tie-breaking used in the second step.
\citet{Dummett1984} suggests that the Borda score\footnote{For each voter, their lowest-ranked candidate is given 0 points, their second-lowest candidate 1 point, their third lowest 2 points, etc. The candidate with the highest total score is chosen.} be used.
\citet{AZIZ2022248} showed that \psc is satisfied regardless of the tie-breaking method.

We show that MD rules fail committee monotonicity whenever positional scoring is used for tie-breaking.
A positional scoring rule consists of a score vector $(s_1, s_2, \ldots, s_m)$, where $s_1 \geq s_2 \geq \ldots \geq s_m \geq 0$ and $s_1 > s_m$. 
A candidate earns $s_r$ points if a voter ranks them in position $r$.
The tie-breaking in step 2 selects the candidate with the highest total score across all $n$ voters. Further tie-breaking may be needed if there are ties in the positional scoring, but our impossibility result holds regardless of the additional tie-breaking method.

\begin{proposition}
	Minimal Demand (MD) rules that break ties using positional scoring do not satisfy committee monotonicity, even for strict preferences.
\end{proposition}
\begin{proof}
	Consider an MD rule with a positional scoring vector $(s_1, s_2, \ldots, s_m)$.
	We assume $s_m = 0$: if it wasn't, then we could use a new vector $(s_1 - s_m, s_2-s_m, \ldots, 0)$ without changing any tie-breaking decisions.
	We consider two cases depending on the scoring vector.
	
	\paragraph{Case 1: $s_1 < 2 s_2$.}
	For this case, we construct a profile $R$ with $n = 2$ voters and $m = 3$ candidates.
	To increase $n$, the entire profile can be duplicated multiple times.
	To increase $m$, extra candidates can be added to the preference relations of the voters as shown below.
	
	\begin{center}
	 	Voter $1$: $c_1 \spref c_3 \spref \text{ any extra candidates } \spref c_2$\\
	 	Voter $2$: $c_2 \spref c_3 \spref \text{ any extra candidates } \spref c_1$\\
	\end{center}
	
	For $k = 1$, a solid coalition must include both voters to be entitled to representation and so no candidates are elected until $r = m$. 
	The tie-breaking will select candidate $c_3$ since its score of $2 s_2$ exceeds the scores of the other candidates.
	For $k = 2$, voters $1$ and $2$ both separately form solid coalitions when $r = 1$, so candidates $c_1$ and $c_2$ are selected.
	This violates committee monotonicity because $c_3$ is elected when $k = 1$ but not when $k = 2$.
	
	\paragraph{Case 2: $s_1 \geq 2 s_2$.}
	For this case, we construct a profile $R$ with $n = 11$ voters and $m = 12$ candidates.
	To increase $n$, the entire profile can be duplicated $\ell$ times.
	To increase $m$, extra candidates can be added to the end of preference relations of each voter.
	\begin{alignat*}{8}
	 	\text{Voter $1$:} \quad           & c_1 & \,\spref\, & c_2   && \spref c_3    && \spref c_4    && \spref c_9    && \spref & \text{ other candidates }\\
	 	\text{Voter $2$:} \quad           & c_2 & \spref & c_3   && \spref c_4    && \spref c_1    && \spref c_9    && \spref & \text{ other candidates }\\
		\text{Voter $3$:} \quad           & c_3 & \spref & c_4   && \spref c_1    && \spref c_2    && \spref c_9    && \spref & \text{ other candidates }\\
		\text{Voter $4$:} \quad           & c_4 & \spref & c_1   && \spref c_2    && \spref c_3    && \spref c_9    && \spref & \text{ other candidates }\\
		\text{Voter $5$:} \quad           & c_5 & \spref & c_6   && \spref c_7    && \spref c_8    && \spref c_9    && \spref & \text{ other candidates }\\
		\text{Voter $6$:} \quad           & c_6 & \spref & c_7   && \spref c_8    && \spref c_5    && \spref c_9    && \spref & \text{ other candidates }\\
		\text{Voter $7$:} \quad           & c_7 & \spref & c_8   && \spref c_5    && \spref c_6    && \spref c_9    && \spref & \text{ other candidates }\\
		\text{Voter $8$:} \quad           & c_8 & \spref & c_5   && \spref c_6    && \spref c_7    && \spref c_9    && \spref & \text{ other candidates }\\
		\text{Voters $9$ to $11$: } \quad &     &        & c_9   && \spref c_{10} && \spref c_{11} && \spref c_{12} && \spref & \text{ other candidates }\\
	\end{alignat*}
	Note that voters $1$ to $4$ form a solid coalition supporting $\{c_1, c_2, c_3, c_4\}$ and voters $5$ to $8$ form a solid coalition supporting $\{c_5, c_6, c_7, c_8\}$.
	For $k = 2$, the committee will consist of one candidate from $\{c_1, c_2, c_3, c_4\}$ and another from $\{c_5, c_6, c_7, c_8\}$.
	Now consider $k = 1$.
	A solid coalition needs at least 6 voters to be entitled to representation, and any such solid coalition will include $c_9$ in its supported candidates.
	We will show that $c_9$ has a higher score than all of $c_1$ through $c_8$, meaning that committee monotonicity will be violated: 
	\begin{itemize}
		\item Candidates $c_1$ through $c_8$ are ranked first, second, third, and fourth by exactly one voter each. Therefore their scores are upper bounded by $s_1 + s_2 + s_3 + s_4 + 7s_5 \leq s_1 + 3s_2 + 7s_5$.
		\item Candidate $c_9$ has a score of $3s_1 + 8s_5$.
	\end{itemize}
	If $s_2 = 0$, then $s_i = 0$ for all $i \geq 2$ and $c_9$'s score of $3s_1$ is higher than $c_1$ through $c_8$'s score of $s_1$. Otherwise, assume that $s_2 > 0$: 
	\begin{align*}
	 	3s_1 + 8s_5 & \geq s_1 + 4s_2 + 8s_5 & \text{ since $s_1 \geq 2 s_2$}\\
		& > s_1 + 3s_2 + 8s_5 & \text { since $s_2 > 0$}\\
		& \geq s_1 + 3s_2 + 7s_5,
	\end{align*}
	and so $c_9$ has a higher score than $c_1$ through $c_8$.
\end{proof}
Since the Borda rule is a positional scoring rule, we have the following corollary. 
\begin{corollary}
	QBS does not satisfy committee monotonicity.
\end{corollary}

\section{Reverse Sequential Ordered Phragmén Fails Candidate Monotonicity}

A committee voting rule is \emph{candidate monotonic} if for every profile $R$, every candidate $c$ and committee size $k$, if $c \in f(R, k)$, and the profile $R'$ is obtained from $R$ by moving $c$ up in some voters' rankings without changing anything else, then $c \in f(R', k)$. \citet[Theorem 14.3]{Jans16a} proves that the Ordered Phragmén rule satisfies this property. However, as we now show, the reverse sequential variant of this rule fails it.

\begin{proposition}
	\label{prop:rev-seq-phragmen-fails-candidate-monotonicity}
	The reverse sequential rule of the Ordered Phragmén rule fails candidate monotonicity.
\end{proposition}
\begin{proof}
Consider the following profile:
\begin{center}
	\begin{tikzpicture}
		\node at (0,0) [draw=black!30] (step1) {
			\votermultiplicity{$v_1$}{\weakorder{{a},{b},{c},{d}}}
			\votermultiplicity{$v_2$}{\weakorder{{a},{b},{c},{d}}}
			\votermultiplicity{$v_3$}{\weakorder{{a},{d},{b},{c}}}
			\votermultiplicity{$v_4$}{\weakorder{{c},{a},{b},{d}}}
			\votermultiplicity{$v_5$}{\weakorder{{d},{c},{a},{b}}}
		};
	\end{tikzpicture}
\end{center}
On this profile, reverse sequential Ordered Phragmén selects the committees $\{a,b,d\}$, followed by $\{a,b\}$, followed by $\{a\}$. In particular, $b$ is a winner for $k = 2$.

The following profile is obtained by moving up $b$ in the ranking of voter 3:
\begin{center}
	\begin{tikzpicture}
		\node at (0,0) [draw=black!30] (step1) {
			\votermultiplicity{$v_1$}{\weakorder{{a},{b},{c},{d}}}
			\votermultiplicity{$v_2$}{\weakorder{{a},{b},{c},{d}}}
			\votermultiplicity{$v_3$}{\weakorder{{b},{a},{d},{c}}}
			\votermultiplicity{$v_4$}{\weakorder{{c},{a},{b},{d}}}
			\votermultiplicity{$v_5$}{\weakorder{{d},{c},{a},{b}}}
		};
	\end{tikzpicture}
\end{center}

On this profile, reverse sequential Ordered Phragmén selects the committees $\{a,c,b\}$, followed by $\{a,c\}$, followed by $\{a\}$. In particular, $b$ has stopped being a winner for $k = 2$.

Thus, this example shows that the reverse sequential Ordered Phragmén rule fails candidate monotonicity.
\end{proof}

\section{Proof of Proposition~5.4}\label{app:scrterm}

\scrterm*
\begin{proof}
	First, we will show that if the SCR rule terminates, it produces a committee of size $k$. To this end, we note that each iteration of the outer loop (\cref{line:outerloop}) adds exactly one candidate to $W$. The reason for this is that the set $\Phi(R,W,D)$ only contains pairs $(N',C')$ such that $C'\not\subseteq W'$, which means that $D$ will always contain at least one alternative not in $W$. On the other hand, we only reach \cref{line:add-candidate} after the end of the while-loop, which requires that $|D\setminus W|\leq 1$. In combination, this means that $|D\setminus W|=1$, so each iteration of the outer loop adds a single candidate.
	
	Next, we will show that SCR always terminates. To this end, we note that if $\Phi(R,W,D)$ is always non-empty during the execution of SCR, then each iteration of the while-loop (\cref{line:whileloop}) reduces the size of $D$ by at least $1$. This holds because $C'$ is a strict subset of $D$ for all $(N',C')\in \Phi(R,W,D)$. Now, consider the set $D$ in some round during the execution of the SCR rule and suppose that $|D\setminus W|>1$. We first note that $D$ is supported by a generalized solid coalition $N''$. In more detail, if $D=C$, this holds as $C$ is supported by the set of all voters $N$. On the other hand, if $D\neq C$, then $D$ was chosen as the set of candidates supported by a generalized solid coalition $N''$ in the previous iteration of the while-loop. Next, let $i\in N''$, which means that $d\succsim_i c$ for all $d\in D$, $c\in C\setminus D$. Moreover, let $d^*$ denote one of voter $i$'s least favorite candidates in $D\setminus W$ and define $D'=\{c\in D\setminus \{d^*\}\colon c\succsim_i d^*\}$. By the definition of this set and $d^*$, it holds that $D'\subsetneq D$,  $D'\not\subseteq W$, and that $\{i\}$ is a generalized solid coalition supporting $D'$. Hence, $(\{i\}, D')\in \Phi(R,W,D)$, which proves that $\Phi(R,W,D)\neq\emptyset$. Thus, the SCR rule is indeed well-defined.
	
	Finally, we show that the SCR rule can be computed in polynomial time for strict preference profiles.
	Since \cref{line:whileloop} is repeated at most $mk$ times, we only need to show that the procedure in the while-loop can be done in polynomial time.
	For this, we note that for strict preferences, \ipsc coincides with \psc and a generalized solid coalition is equivalent to a solid coalition.
	Furthermore, there are only $nm$ candidate sets $C'$ such that $(N',C')$ can be in $\Phi(R,W,D)$: for each voter $i$ and each rank $r$, we can obtain one such set by considering the $r$ most preferred candidates of voter $i$. 
	Finally, we can compute \cref{line:choose-N*-and-C*} by evaluating $\rho(W, N', C')$ for all such sets $C'$ and the maximal solid coalitions $N'$ supporting $C'$. This works because for any two solid coalitions $N'$, $N''$ supporting $C'$, it holds that $\periphery{C'}{N'}=\periphery{C'}{N''}$ and thus $\rho(W, N', C')\geq \rho(W, N'', C')$ if $|N'|\geq |N''|$. Moreover, we can compute the maximal solid coalition supporting a set $C'$ by simply checking for each voter whether he supports $C'$. Since all of these steps can be done in polynomial time, it follows that the SCR rule can be computed in polynomial time for strict preferences. 
\end{proof}

\section{Proof of Theorem 8.3}\label{app:swap}

\swapdistance*
\begin{proof}
	We provide here the proof for the claim on Droop-\psc, which works analogously to the proof for Hare-\psc in the main body. Specifically, we fix again a profile $R$, a group of voters $N'$ that all report the same preference relation $\succsim$ in $R$, and we let $\alpha\in[0,1]$ such that $|N'|\geq \alpha\cdot n$. Since $N'$ forms a solid coalition for every set $\{c_1,\dots,c_r\}$ with $r\in \{1,\dots, m\}$, Droop-PSC implies that all candidates $c_1,\dots, c_r$ are chosen when the committee size $k$ satisfies that $\alpha>\frac{r}{k+1}$. Next, let $\rhdsim_D$ denote a ranking that satisfies Droop-PSC for $R$. Solving our previous inequality for $k$ shows that, for all $r$, the candidates $\{c_1,\dots,c_r\}$ are among the top-$\ell$ in every ranking that satisfies Droop-PSC, where $\ell$ is the smallest rank larger than $\frac{r}{\alpha} - 1$. Now, let $\varepsilon > 0$ be fixed. Then, based on our previous considerations, we know that the candidates $c_1, \dots, c_r$ are all ranked in the first $\lceil \frac{r}{\alpha} + \varepsilon\rceil - 1$ candidates. Let $x$ be a candidate ranked by $\rhdsim_D$ at rank $t$ with $\lceil \frac{r}{\alpha} + \varepsilon\rceil \le t \le \lceil \frac{r+1}{\alpha} + \varepsilon\rceil - 1$. By Droop-PSC, it must be that $x=c_i$ for some $i>r$, so candidate $x$ can have at most $t-r-1$ inversions to candidates ranked before it in $\rhdsim_D$. The last rank $r'$ with $\lceil \frac{r'}{\alpha} + \varepsilon\rceil - 1 \le m-1$ is now $r' = \lfloor \alpha (m - \varepsilon)\rfloor$. Using this and the same technique as for the Hare quota, we derive the following inequality:
\begin{align*}
	\swap(\succsim, \rhdsim_D) &\le \sum_{i = 0}^{\lceil 1/\alpha + \varepsilon\rceil - 2}(i) + \sum_{i = \lceil 1/\alpha + \varepsilon\rceil -1}^{\lceil 2/\alpha + \varepsilon\rceil - 2}(i - 1) + \dots +  \sum_{i = \lceil r/\alpha + \varepsilon\rceil -1}^{\lceil (r+1)/\alpha + \varepsilon\rceil - 2}(i - i) + \dots  \sum_{i = \lceil r'/\alpha + \varepsilon \rceil-1}^{m-1}(i - r') \\
	& \le \sum_{i = 0}^{m - r' - 1} (i) + \sum_{i = 1}^{r'} (\lceil i/\alpha + \varepsilon\rceil - i - 1)\\
	& \le (m - r')(m - r'-1)/2  + \sum_{i = 1}^{r'} ( i/\alpha - i + \varepsilon)\\
	& =  (m - r')(m - r'-1)/2 + (1/\alpha - 1)(r'+1)(r')/2 + \varepsilon r'.\\
\end{align*}
Just as in the previous case, this is a quadratic function in terms of $r'$. This quadratic function ($(r')^2/(2\alpha) + (\varepsilon + \frac{2}{\alpha} - m)r' + m^2/2 - m/2)$) has a minimum point at $\alpha( m - \varepsilon) - \frac{1}{2}$ which again allows us to upper bound the swap-distance by substituting $r'$ with $\alpha (m - \varepsilon)$, which gives us
\begin{align*}
	\swap(\succsim, \rhdsim_D) & \le (m - \alpha(m - \varepsilon))(m - \alpha(m - \varepsilon)-1)/2\\ 
	&\quad + (1/\alpha - 1)(\alpha(m - \varepsilon)+1)(\alpha(m - \varepsilon))/2 + \varepsilon \alpha(m - \varepsilon)\\
	&= \frac{\alpha \varepsilon^2}{2} - \frac{\varepsilon}{2} - \frac{\alpha m^2}{2} + \frac{m^2}{2} + \varepsilon \alpha (m - \varepsilon). 
\end{align*}
For $\varepsilon \to 0$, this term converges to the desired bound of $((1 - \alpha) \frac{m}{m-1})\cdot {m\choose 2}$.
\end{proof}

\end{document}